%% file: CPP_paper_v2.tex
\documentclass[sigplan,screen]{acmart}
\input{CPP_preamble.tex}

\AtBeginDocument{%
  }

\setcopyright{acmlicensed}
\acmDOI{10.1145/3636501.3636942}
\acmYear{2024}
\copyrightyear{2024}
\acmSubmissionID{poplws24cppmain-p13-p}
\acmISBN{979-8-4007-0488-8/24/01}
\acmConference[CPP '24]{Proceedings of the 13th ACM SIGPLAN International Conference on Certified Programs and Proofs}{January 15--16, 2024}{London, UK}
\acmBooktitle{Proceedings of the 13th ACM SIGPLAN International Conference on Certified Programs and Proofs (CPP '24), January 15--16, 2024, London, UK}
\received{2023-09-19}
\received[accepted]{2023-11-25}

\begin{document}


\title{A Formalization of Complete Discrete Valuation Rings and Local Fields}


\author{María Inés de Frutos-Fernández}
\orcid{0000-0002-5085-7446}
\authornote{Both authors contributed equally to this research.}
\affiliation{%
	\department{Departamento de Matem\'aticas}
	\institution{Universidad Aut\'onoma de Madrid}
	\city{Madrid}
	\country{Spain}
}
\email{maria.defrutos@uam.es}

\author{Filippo Alberto Edoardo Nuccio Mortarino Majno di Capriglio}
\orcid{0000-0002-5318-9869}
\authornotemark[1]
\affiliation{%
	\department{Institut Camille Jordan UMR 5208}
	\institution{Université Jean Monnet Saint-Étienne}
	\city{Saint-Étienne}
	\country{France}
}
\email{filippo.nuccio@univ-st-etienne.fr}

\renewcommand{\shortauthors}{M.~I.~de Frutos-Fernández and F.~A.~E.~Nuccio Mortarino Majno di Capriglio}

\begin{abstract}
  Local fields, and fields complete with respect to a discrete valuation, are essential objects in commutative algebra, with applications to number theory and algebraic geometry. We formalize in Lean the basic theory of discretely valued fields. In particular, we prove that the unit ball with respect to a discrete valuation on a field is a discrete valuation ring and, conversely, that the adic
  valuation on the field of fractions of a discrete valuation ring is discrete. We define finite extensions of valuations and of discrete valuation
  rings, and prove some localization results. 
  
  Building on this general theory, we formalize the abstract definition and some fundamental properties of local fields. As an application, we show that finite extensions of the field $\QQ_p$ of $p$-adic numbers and of the field $\laurentseries$ of Laurent series over $\FF$ are local fields. 
  
\end{abstract}




\maketitle

\section{Introduction}\label{sec:intro}
The vague idea that geometric intuition and algebraic rigor can fruitfully interact is an old theme, certainly predating the modern approach by Descartes and subsequently by Newton and Leibniz. 
In contemporary commutative algebra, and especially
after the advent of scheme theory by Grothendieck around 1960, this connection has become even tighter: geometric concepts and techniques are often borrowed for a wide range of applications. In this work we are concerned with the concept of \emph{local field}, a fundamental notion in algebraic number theory whose origin, and still many applications, comes from geometry and diophantine questions. Suppose, for instance, that one is interested in the integral solutions $(a_1,a_2,a_3)$ to $X_1^2+X_2^2+X_3^2=0$: clearly, since for positivity reasons the only \emph{real} solution is $(0,0,0)$, there cannot be any other integral solutions. On the other hand, $X_1^2+ X_2^2-3X_3^2=0$ certainly has real solutions, yet has no non-trivial solutions in $\FF[3]$, the field with $3$ elements, as can be seen by analyzing the finitely many possibilities. Hence, if $(a_1,a_2,a_3)$ is an integral solution to this equation, its reduction modulo 3 must be the trivial solution in $\FF[3]$, so that $a_1, a_2$ and $a_3$ are all divisible by $3$. This implies $a_1 = a_2 = a_3 = 0$, since otherwise we would have an equality $a_1^2 + a_2^2 = 3a_3^2$, in which the highest power of $3$ dividing the left hand side is even, while the highest power of $3$ dividing the right hand side is odd, yielding a contradiction.

From a geometric perspective, one can regard the previous as a kind of ``local analysis'': interpreting the primes as the points of a geometric object attached to $\ZZ$, if a ``global'' solution $(a_1,a_2,a_3)\in\ZZ^3$ exists, then for each prime $p$, we obtain a solution $(\overline{a_1},\overline{a_2},\overline{a_3})\in\FF^3$ by reducing the coordinates modulo~$p$ --- we can look at this as being a local solution ``around the point~$p$'' --- as well as the solution  $((a_1 : \RR),(a_2 : \RR),(a_3~:~\RR))\in\RR^3$ over the reals. In particular, if one of these local solutions fails to exists, this means that the equation cannot have a global solution. Yet, in this analogy, the fields $\FF$ and $\RR$ are very different: the former is finite, with trivial discrete geometry and of positive characteristic, while the latter has a rich metric structure and characteristic $0$. Now, for every $p$, one can rather consider the field of $p$-adic numbers $\QQ_p$ which is, in this perspective, a better analogue for $\RR$ than $\FF$: it has a metric with respect to which it is complete (Cauchy sequences converge), and it has characteristic $0$ (so, in particular, it contains $\QQ$). To define it, one observes that for every prime $p$ there exists an absolute value $\lvert\cdot\rvert_p$ on $\QQ$ for which numbers that are highly divisible by $p$ are close to $0$. Having an absolute value, it is possible to speak about convergence and about Cauchy sequences: exactly as for the euclidean absolute value, one can find Cauchy sequences that do not converge to any rational, and $\QQ_p$ is defined as the \emph{completion} of $\QQ$ with respect to~$\lvert\cdot\rvert_p$. It is the ``smallest'' field where all $\lvert\cdot\rvert_p$-Cauchy sequences converge. Crucially, it still bears strong connections with the prime $p$: for instance, if a monic polynomial in $\ZZ[X]$ has a simple root in $\FF$, then it also has one in $\QQ_p$. 

The local analysis described before is an example of the ``local-to-global'' principle: to solve a problem in $\ZZ$, or $\QQ$, one can try to solve it in $\QQ_p$, for all $p$, and in $\RR$. A famous application of the local-to-global principle is the Hasse–Minkowski Theorem, which states that if $Q(x)$ is a quadratic form with rational coefficients, then the equation $Q(x)=0$ has a nontrivial solution over $\QQ$ if and only if it has a nontrivial solution over $\RR$ and over every $\QQ_p$.

For another notorious example, it has long been known that the so-called ``first case'' of Fermat's Last Theorem can be solved using class field theory (see~\cite{LenSte97}), which is a deep theory where the global theorems are obtained by a beautiful patching of the local ones: we refer to \cite{Tat67}~for an excellent account of this example. It should also be emphasized that ultimately Wiles' very proof of Fermat's Last Theorem makes heavy use of the local-to-global principle in the framework of Galois representations: we refer to~\cite{Fre09} for a summary of this technique.

Before describing our work, let us mention two generalizations of the $p$-adic numbers that will guide our approach to the formalization. First of all, the field $\QQ$ of rational numbers can be replaced by an arbitrary number field, or global field. We refer to~\cite[\S2]{BaaDahNarNuc22} for a brief overview of these notions in the context of formalization. Number fields are fields that can be obtained by adjoining roots of polynomials $f(X)\in\ZZ[X]$ to $\QQ$, like $\QQ(i)$ or $\QQ(\sqrt{2})=\{r+s\sqrt{2}, r,s\in\QQ\}$. They are a key object of study in algebraic number theory and resemble in many perspectives the rational field. Every number field $K$ can be realized as the field of fractions $K=\Frac{\rint{K}}$ of a suitable subring $\rint{K}\subseteq K$, in the same vein as $\QQ=\Frac{\ZZ}$. These rings $\rint{K}$ are not principal, in general, but they are Dedekind domains (see Example~\ref{example:adic_valuation} and the references~\ibid): so, although we cannot in general define an absolute value $\lvert\cdot\rvert_\pi$ on $\rint{K}$ associated to elements $\pi\in\rint{K}$, it is possible to define an absolute value $\lvert\cdot\rvert_{\maxid[\empty]}$ associated to every maximal ideal of $\rint{K}$. The completion procedure discussed above can then be performed analogously, and the fields $\compl[{\maxid[\empty]}]{K}$ obtained in this way are all finite extensions of some $\QQ_p$, and are the so-called mixed characteristic local fields. Analogous constructions exist when replacing $\QQ$ by $\FF(X)$, leading to equal characteristic local fields, and we refer to~\S\ref{subsec:eq_char} for more details. There is also a more abstract definition of local fields, that includes the above two, that is recalled and formalized in Definition~\ref{def:local_field}.

Secondly, unlike the euclidean metric, the $p$-adic one carries a strong discrete flavor, in that only integral powers of~$p$ can occur as values $\lvert x\rvert_p$ for $x\in\QQ\setminus\{0\}$. This happens in other contexts, and the relevant notion is that of a discrete valuation ring; we will describe their theory in more detail in~\S\ref{subsection:def_dvr}. Here we content ourselves with saying that they are a class of rings~$R$ endowed with a valuation map $v$ to $\ZZ\cup\{\infty\}$ such that the arithmetic behavior of elements $r\in R$ can be read on the value $v(r)$. The relation of discrete valuation rings with local fields is deep: the completion of the fraction field of a discrete valuation ring with finite residue field (see~\S\ref{sec:local_fields}) is a local field and, conversely, for every local field $F$ the subset $\{x \in F \:\vert\: \lvert x\rvert_F\leq 1\}$ is always a discrete valuation ring, yet not every discrete valuation ring is of this form.

In this paper, we formalize in \lean the definition and basic theory of local fields and their relation with discrete valuation rings. We provide both a concrete approach, defining mixed characteristic and equal characteristic local fields as finite field extensions of $\QQ_p$ and $\laurentseries$, and a more abstract definition of local field which comprises these two special cases. 

The \mathlib library provides 
a very complete theory of general valuations, as well as the basic theory of discrete valuation rings, but the connection between these two notions is largely missing. We formalize the definition of discrete valuation and greatly expand the available theory for discrete valuation rings, their fields of fractions, and their associated discrete valuations, including theorems about completions, extensions of valuations and localizations.
To the authors' best knowledge, this work constitutes the first formalization of local fields and of extensions of discrete valuations in any proof assistant.

The \lean[\empty] code for our formalization is available at a public GitHub repository\href{https://github.com/mariainesdff/local_fields_journal}{\extlink}. Inside the folder \code{from_mathlib}\href{https://github.com/mariainesdff/local_fields_journal/tree/master/src/from_mathlib}{\extlink} we gather files that are previous work of the authors, as well as a file due to Yaël Dillies \href{https://github.com/mariainesdff/local_fields_journal/blob/master/src/from_mathlib/PR18604_well_founded.lean}{\extlink}, used in this project after securing permission from them. Everything else is new, original work of the authors of this paper, totaling about 8,000 lines of code. Some of the \LClistingname s below have been edited for clarity, but we also provide links to the corresponding results in our repository. In some \LClistingname s the so-called \emph{dot notation} is exhibited: for example, the call \code{(ideal.is_prime I)} can be shortened to \code{I.is_prime}.

Throughout this paper, all rings are assumed to be commutative and unitary.

\subsection{\texorpdfstring{\lean[\empty] and \mathlib}{Lean and mathlib}}
This project is formalized in the \lean theorem prover \cite{Lean3}, which is based on dependent type theory, with proof irrelevance, quotient types, and non-cumulative universes \cite{TypeTheory}.

We build our work on top of the mathematical library \mathlib of \lean{\empty}, whose key properties are its unified approach to integrate different mathematical theories, and its decentralized nature, with over 300 contributors. One of the key tools for organizing mathematical hierarchies used both in \lean[\empty]\kern-.5em's core library and in \mathlib is typeclass inference. By declaring certain terms as \code{instance}, \lean[\empty] will automatically try to infer the relevant values during the unification procedure. We refer to~\cite{Baa22} for a more detailed discussion of the role of \code{instance} parameters in \lean[\empty]\kern-.5em.

The formalization is in \lean because at the time when we began it, almost none of its prerequisites were available in \lean[4]. Now that the complete \mathlib has been translated to \lean[4], we plan to port our work and start to integrate it in this library. All our future work related to this project will be directly implemented in \lean[4].

\subsection{Paper outline}
In \S\ref{section:dvr}, we treat the general theory of discrete valuations and discrete valuation rings: after presenting some background material in \S\ref{subsection:wzmZZ}--\S\ref{subsection:def_dvr}, we describe our main results in \S\ref{subsec:complete_DVR}, concerning completions, and \S\ref{subsec:extensions}, concerning extensions of valuations. In \S\ref{def:local_field} we present our formalization of local fields, describing the equal and mixed characteristic cases in \S\ref{subsec:eq_char} and \S\ref{subsec:mixed_char}, respectively. We conclude in \S\ref{sec:conclusion} with a description of future and related work and a discussion of some of our key design choices.

\section {Discrete Valuation Rings}\label{section:dvr}

In this section we present our formalization of the theory of discretely valued fields and discrete valuation rings, including results about completions and extensions of valuations. As a technical prerequisite, we discuss in \S~\ref{subsection:wzmZZ} the definition and main properties of the \mathlib type $\wzmZZ$, where most of the valuations that we considered take their values.

\subsection{\texorpdfstring{The type $\wzmZZ$}{The type with\_zero multiplicative Z}}\label{subsection:wzmZZ}
There are several situations in mathematics in which additive structures get translated into multiplicative structures. For example, associated to each real number $n \ne 0$, there is an exponential map $\exp_n : \RR \to \RR$ sending $x$ to $n^x$. This map has the property that
$\exp_n(x + y) = \exp_n(x) \cdot \exp_n(y)$ and, whenever $n > 1$, it preserves the order on $\RR$.

A second example occurs in the theory of valuations. If $\av : R \to \ZZ \cup \{\infty\}$ is an additive valuation on a ring $R$ (see Definition~\ref{def:add_valuation}) and $n > 1$ is a real number, it is common to study the associated function $v \colon R \rightarrow n^{\ZZ} \cup \{0\}$ sending $r$ to $n^{-\av(r)}$, with the convention that $n^{-\infty} = 0$. Here
\[ n^{\ZZ} \cup \{0\}=\{n^a \:\vert\: a \in \ZZ\} \cup \{0\}\subseteq \RR. \]

Some abstractions that can be used to formalize this kind of translations are available in \mathlib. Given any type \code{A} endowed with an additive structure, \mathlib provides a new type \code{multiplicative A} that is in bijection with \code{A} and carries a multiplicative structure, together with a map
\code{of_add} : \code {A} $\rightarrow$ \code{multiplicative A} satisfying
\[
    \mcode{of\_add\; (x + y) = of\_add\; x\; *\; of\_add\; y}
\]
for all \code{x, y : A}. If \code{A} is equipped with a preorder, then so is \code{multiplicative A}, and the map \code{of_add} is strictly monotone: \code{x} $\le$ \code{ y} is equivalent to \code{of_add x} $\le$ \code{of_add y}. The map \code{to_add : multiplicative A} $\to$ \code{A} is inverse to \code{of_add}.

We are especially interested in the case \code{A} $= \ZZ$. As an additive group, $\ZZ$ is an infinite cyclic group generated by~$1$ or~$-1$. Correspondingly, the elements \code{of_add(1)} and \code{of_add(-1)}, that do not bear a specific notation, are the only generators of the cyclic group \code{multiplicative}\;$\ZZ$. Since the map \code{of_add} preserves the unit element, we have that \code{of_add(0) = 1} is the unit, hence
\[
\text{\code{(1:}}\;\text{\code{multiplicative} }\ZZ)\neq\text{\code{of_add(1:}}\ZZ\text{\code{)}}.
\]
Given two integers $a, b$, it holds that \code{of_add(a * b) = of_add(a)}$^b$\code{= of_add(b)}$^a$. In particular, we have the equality \code{of_add(}$\pm$\code{1)}$^n$= \code{of_add(}$\pm$\code{n)} for all \code{(n :} $\ZZ$\code{)}.

A new type \code{with_zero (multiplicative} $\ZZ$\code{)}, denoted $\wzmZZ$, can be constructed from \code{multiplicative} $\ZZ$ by adding an extra term \code{0} to \code{multiplicative} $\ZZ$. The resulting type $\wzmZZ$ still carries a multiplication, that extends the one on \code{multiplicative} $\ZZ$ and for which \code{0 * x = 0} for all \code{x}; moreover, setting \code{0} $\le$ \code{x} for all \code{x}, $\wzmZZ$ is endowed with an order that extends the one on \code{multiplicative} $\ZZ$. The inclusion
\[
\text{\code{multiplicative} }\ZZ \rightarrow \wzmZZ
\]
is an order-preserving coercion that respects the multiplication on both sides, and we omit it from the notation.

The type $\wzmZZ$ provides an abstraction of the structure on the set $n^{\ZZ} \cup \{0\}$ that does not require a choice of the base $n$, and it will be the codomain of most of the valuations that we consider in this paper.

\subsection{Valuations}\label{subsection:valuations}
\begin{definition}[{\cite[Chapitre~VI, \S3, n$^\circ$1, D\'efinition~1]{Bou85}}]\label{def:add_valuation}
Let $R$ be a ring. A function $\av \colon R \to \ZZ \cup \{\infty\}$ is an \emph{additive valuation} on $R$ if
\begin{listDef}
	\item $\av(r) = \infty$ if and only if $r = 0$; \label{def_add_valuation:zero}
	\item $\av (r\cdot s) = \av(r) + \av(s)$ for all $r, s$ in $R$;\label{def_add_valuation:mul}
	\item $\av (r + s) \geq \min \{\av(r), \av(s)\}$ for all $r, s$ in $R$. \label{def_add_valuation:add}
\end{listDef}
\end{definition}

Given an additive valuation $\av$ on $R$, we can define an associated function $v \colon R \rightarrow \wzmZZ$ by sending $0$ to $0$ and $r \ne 0$ to \code{of_add(-}$\av$\code{(r))}. From the definition of $\av$, it is easy to deduce that $v$
satisfies the properties of a multiplicative valuation, as in the following definition:
\begin{definition}\label{def:valuation}
A \emph{multiplicative valuation} on a ring $R$ is a function $v\colon R \rightarrow \wzmZZ$ satisfying the properties
\begin{listDef}
	\item $v (0) = 0$;\label{def_valuation:zero}
	\item $v (1) = 1$;\label{def_valuation:one}
	\item $v (x \cdot y) = v (x) \cdot v (y)$ for all $x, y \in R$;
 \label{def_valuation:mul}
	\item $v (x + y) \le \max \{v (x), v (y)\}$ for all $x, y \in R$.\label{def_valuation:add}
 \end{listDef}

\end{definition}
Note that an element $r \in R$ has additive valuation $1$ if and only if it has multiplicative valuation \code{of_add(-1)}. Elements of multiplicative valuation \code{of_add(-1)} will play a prominent role in \S\ref{subsection:def_dvr}.

For example, let $R = \ZZ$ be the ring of integers, and fix a prime number $p$. Then, thanks to unique factorization, we can define an additive valuation $\av[p] \colon \ZZ \to \ZZ \cup \{\infty\}$ by sending an integer $z$ to the number of times that $p$ appears in the factorization of $z$. We can extend the function $\av[p]$ to $\QQ$ by defining $\av[p](r/s) := \av[p](r) - \av[p](s)$. Then $\av[p]$ is an additive valuation on $\QQ$, called the \emph{additive $p$-adic valuation} on $\QQ$. 
The corresponding multiplicative valuation $v_p \colon \QQ \to \wzmZZ$ is called the \emph{$p$-adic valuation}.


\begin{example}[The {$\maxid[\empty]$}-adic valuation]\label{example:adic_valuation}
We will often consider the following generalization of the $p$-adic valuation. If $R$ is a Dedekind domain (see~\cite[Chapitre~VII, \S2]{Bou85}) which is not a field, then every nonzero ideal of $R$ can be factored as a product of maximal ideals, uniquely up to reordering. Therefore, for every maximal ideal $\maxid[\empty]$ of $R$, we can follow an analogous construction to define an additive valuation $\av[{\maxid[\empty]}] \colon R \to \ZZ \cup \{ \infty\}$ on $R$ by sending $0$ to $\infty$ and any nonzero $r \in R$ to the number of times that $\m$ appears in the factorization of the principal ideal $(r)$. We extend~$\av[{\maxid[\empty]}]$ to the fraction field $K$ of $R$ by the formula $\av[{\maxid[\empty]}](r/s):=\av[{\maxid[\empty]}](r)-\av[{\maxid[\empty]}](s)$. The corresponding multiplicative valuation $v_{\maxid[\empty]}=\mcode{of\_add}\circ (-\av[{\maxid[\empty]}])\colon K \to \wzmZZ$ on $K$ is called the \emph{$\m$-adic valuation}. This valuation was formalized in \cite{deF22}, and is available in \mathlib as the declaration~\code{is_dedekind_domain.height_one_spectrum.valuation}\href{https://leanprover-community.github.io/mathlib_docs/ring_theory/dedekind_domain/adic_valuation.html#is_dedekind_domain.height_one_spectrum.valuation}{\extlink}.
\end{example}

While in most number theory references it is more common to work with additive valuations than with multiplicative ones, for historical reasons \mathlib follows the opposite convention, and it provides a much more complete API for multiplicative valuations. Hence we use multiplicative valuations in our formalization and, throughout the paper, whenever we use the term "valuation" without further adjectives, we mean "multiplicative valuation".

\begin{remark}\label{rmk:value_group}(For experts)
More generally, the codomain $\wzmZZ$ in the definition of a (multiplicative) valuation can be replaced by $\Gamma_0=\Gamma\cup\{0\}$, where $\Gamma$ is a linearly ordered commutative group. The order on $\Gamma$ is extended to $\Gamma_0$ analogously as we did for $\wzmZZ$. Structures like $\Gamma_0$ are called ``groups with zero''.
    
An example of this situation occurs when starting with a \emph{local} Dedekind domain $R$, with unique maximal ideal $\maxid$ and field of fractions $K$. In this setting one can prove that the quotient map
 \[
v \colon K\rightarrow K/R^\times
\]
is a valuation, where the quotient $K/R^\times$ is ordered by reverse divisibility. The group with zero $K/R^\times=K^\times/R^\times\cup\{0\}$ is called the value group\footnote{We translate in this way the term \emph{groupe des valeurs} from~\cite[Chapitre~VI, \S3, n$^\circ$2]{Bou85}.} of $R$. The value group is implemented in \mathlib as \code{valuation_ring.value_group}\href{https://leanprover-community.github.io/mathlib_docs/ring_theory/valuation/valuation_ring.html#valuation_ring.value_group}{\extlink}, but we mostly stick to $\wzmZZ$-valued valuations in our work.
\end{remark}

We represent valuations using the \code{valuation}\href{https://leanprover-community.github.io/mathlib_docs/ring_theory/valuation/basic.html#valuation}{\extlink} structure
available in \mathlib, which encodes Definition \ref{def:valuation}. A valuation on a ring $R$ induces a topology that is homogeneous with respect to addition (see~\cite[Chapitre~III,\S1]{Bou71}: every neighborhood $\Omega\subseteq R$ is of the form $\Omega=x+\Omega_0$ where $\Omega_0$ is a neighborhood of $0$ and $x\in\Omega$. In other words, addition (and subtraction) are homeomorphisms of the ring into itself. This property shows that to characterize the topology it is enough to describe the neighborhoods of $0$.
 
 When $R$ is a Dedekind domain and $v=v_{\maxid[\empty]}$ is the topology associated to a maximal ideal $\maxid[\empty]$, as defined in Example~\ref{example:adic_valuation}, 
 a basis for the neighborhoods of $0$ is provided by the sets
 \[
 U_\gamma = \{r \in R \big\vert\,v_{\maxid[\empty]}(r)\leq \gamma\}\quad\text{ for }\quad\gamma \colon \wzmZZ,\; \gamma \ne 0.
 \]

In particular, two elements $r_1,r_2$ are ``close to each other'' 
if their difference lies in a sufficiently deep neighborhood of~$0$, meaning that $v_{\maxid[\empty]}(r_1-r_2)\leq \gamma$ for a suitable $ \gamma \colon \wzmZZ, \gamma \ne 0$. Given the definition of $v_{\maxid[\empty]}$ this translates into the fact that the principal ideal $(r_1-r_2)$ is divisible by a high power of $\maxid[\empty]$.

Actually, the above topology is of a special kind, because the valuation defines a structure of \emph{uniform space} on $R$, as explained in~\cite[Chapitre~VI, \S5]{Bou85} and in the references \ibid. This is a simultaneous generalization of the structure of metric space and of topological group and the topology is induced by the uniform structure. Our main reference for uniform structures is~\cite[Chapitre~II]{Bou71}; for a throughout discussion of the formalization of uniform spaces, we urge the reader to consult the  beautifully written~\cite{BuzComMas20}, in particular its~\S5.
 
The \mathlib library also provides the class \code{valued}\href{https://leanprover-community.github.io/mathlib_docs/topology/algebra/valuation.html#valued}{\extlink}, which bundles together a ring $R$ endowed with a uniform space structure and a distinguished valuation inducing it. Given a term  (\code{hv : valued R} $\wzmZZ$), these can be accessed through \code{hv.to_uniform_space} and \code{hv.v}, respectively.

This class is designed for rings in which there is a preferred valuation. Recall from Example \ref{example:adic_valuation} that 
on a Dedekind domain $R$, there is a valuation for each maximal ideal, and it can be shown that any nontrivial valuation on $R$ is of this form. Hence, if $R$ is local, then the only nontrivial valuation on $R$ is the  $\maxid$-adic valuation associated to its unique maximal ideal, and we declare a \code{valued} instance on $R$. If $R$ is not local, given any maximal ideal $\maxid[\empty]\subseteq R$  we can still define a term \code{(hv}$_{\maxid[\empty]}$\code{: valued R}\;$\wzmZZ$\code{)} representing the $\maxid[\empty]$-adic valuation and allowing us to access the whole \code{valued} API locally. Yet, we would not declare this as a global \code{valued} instance on $R$, since none of the $\maxid[\empty]$-adic valuations on $R$ is preferred over the others.

Given a ring $R$ with a valuation $v$, we can consider the \textit{unit ball} of $R$, that is the subring $\ball{R}$ of elements with valuation less than or equal to (\code{1:} $\wzmZZ$). This subring is called \code{v.integer}\href{https://leanprover-community.github.io/mathlib_docs/ring_theory/valuation/integers.html#valuation.integer}{\extlink} in \mathlib. 
When $K$ is a field, the subring $\ball{K}$ is a \emph{valuation subring}, meaning that for every $\alpha\in K$, either $\alpha \in \ball{K}$ or $\alpha^{-1} \in \ball{K}$: in particular, $K\cong\Frac(\ball{K})$. Since the definition of valuation subring involves taking inverses, it is only available for fields. Hence, when working with a general ring $R$, we formalize $\ball{R}$ as \code{v.integer}, but when working with a field we use the richer structure \code{v.valuation_subring}\href{https://leanprover-community.github.io/mathlib_docs/ring_theory/valuation/valuation_subring.html#valuation_subring}{\extlink}, which gives us access to results about valuation subrings available in~\mathlib. 

\subsection{Discrete Valuation Rings}\label{subsection:def_dvr}

Our general references for the theory of discrete valuation rings are~\cite[Chapitres~I--II]{Ser62} and~\cite[Chapitre~VI]{Bou85}. One notable difference with our language is that both references consider additive valuations, whereas we stick to the multiplicative convention introduced in \S\ref{subsection:valuations}, which is the approach chosen in \mathlib. 
Mathematically, the two points of view are equivalent: one just needs to keep in mind the translation between $\ZZ \cup \{\infty\}$ and $\wzmZZ$. Finally, we refer to~\cite{Bou07}, in particular to \cite[Chapitre~IV]{Bou07}, for the main results about ring theory and commutative algebra that we will need.

Let $R$ be a ring, as above assumed commutative and unitary. We begin by recalling the following equivalence:
\begin{theorem}[{\cite[Chapitre~VI, \S3, n$^\circ$6, Proposition~9]{Bou85}}]\label{thm:dvr_tfae}
Suppose $R$ is a Noetherian local ring that is not a field. The following properties are equivalent:
\begin{listResults}
\item The maximal ideal $\maxid$ of $R$ is principal;\label{dvr_tfae:generator}
\item $R$ is a principal ideal domain \textup{(}PID\textup{)};\label{dvr_tfae:PID}
\item $R$ is an integral domain that coincides with the unit ball of its fraction field $\Frac(R)$ with respect to a valuation $v\colon \Frac{R}\to\wzmZZ$.\label{dvr_tfae:pts_val}
\end{listResults}
\end{theorem}
\begin{remark}
In~\cite[Chapitre~VI, \S3, n$^\circ$6, Proposition~8]{Bou85} and in~\cite[Chapitre I, \S2, Proposition~2]{Ser62} other equivalences are proved, but we will not need them. 
\end{remark}

A ring satisfying the equivalent properties of Theorem~\ref{thm:dvr_tfae} is called a discrete valuation ring (often shortened as DVR). The nomenclature is motivated by point~\ref{dvr_tfae:pts_val}~\ibid: given a DVR $R$, it is possible to find a valuation
\[
v\colon \Frac(R)\rightarrow \wzmZZ
\]
such that $R=\ball{(\Frac R)}$ : yet this valuation is not unique. Indeed, the definition of the action of $\ZZ$ on $\wzmZZ$ shows that replacing a valuation $v$ by a power $v^e$ for $e\in \ZZ_{>0}$ leaves the unit ball unchanged. For each of these valuations, the image $\image(v)\subseteq\wzmZZ$ is the free group with zero generated by $v(r)$, where $r$ is any generator of the maximal ideal $\maxid$ of $R$. Upon replacing $v$ by $v^{1/a(r)}$ for some generator $r$ of $\maxid$, where $\av$ denotes the additive valuation associated to $v$, we can assume that the image is the whole $\wzmZZ$: in this case the valuation is said to be \emph{normalized}, and the elements of valuation equal to (\code{of_add}$(-1) : \wzmZZ$) are called \emph{uniformizers}.

One basic result is that, for a normalized valuation on a DVR, the uniformizers are exactly the set of generators of $\maxid$, because properties~\ref{def_valuation:one} and \ref{def_valuation:mul} of Definition~\ref{def:valuation} show that an element \mbox{$u\in R$} is a unit if and only if $v(u)=1=\;$(\code{of_add}$\mcode{0}\!)$. In particular, there exists a unique normalized valuation on a DVR, since the same argument shows that every valuation is uniquely determined by its value on one, or any, generator of $\maxid$. Moreover, as explained in Example~\ref{example:adic_valuation}, any such valuation can uniquely be extended to $\Frac(R)$. We denote this normalized valuation by ${\normalised}_{,R}$ --- or simply by $\normalised$ when the DVR $R$ can be understood.

Assume now that $R$ is any integral domain (not necessarily a DVR) and that $K$ is a field of fractions for $R$. The points of view taken in~\cite[Chapitre~I]{Ser62} and in~\cite[Chapitre~VI, \S3, n$^\circ$6]{Bou85} concerning valuations on $K$ are not identical: in the former, the valuations occurring in point~\ref{dvr_tfae:pts_val} of Theorem~\ref{thm:dvr_tfae} take values in $\ZZ\cup\{\infty\}$ (or $\wzmZZ$, up to our translation), whereas in the latter they are valued in $K^\times/R^\times \cup\{\infty\}$, that translates to the value group $K/R^\times=K^\times/R^\times\cup\{0\}$ in our language, but under the assumption that there is an isomorphism $K^\times/R^\times\cong \ZZ$. In both cases they are called ``discrete''; they are said to be normalized, as above, when they are surjective. Often results are stated assuming that the valuation is normalized, relying on the possibility of achieving this normalization simply by rescaling, so that the theory of discrete valuations is essentially the theory of \emph{normalized} discrete valuations.

In the setting of a formalization work, the ambiguities described above concerning the group (with zero) $\Gamma_0$ and the normalization of the valuation call for more attention than in pen-and-paper mathematics. It is in principle possible to let $\Gamma_0$ vary, including it as a (perhaps implicit) variable: yet one must stipulate the existence of an isomorphism $\Gamma_0\cong\wzmZZ$, adding one more variable, so that we rather find it more convenient to consider only valuations that are $\wzmZZ$-valued, in line with the choice made when defining the adic valuation\href{https://leanprover-community.github.io/mathlib_docs/ring_theory/dedekind_domain/adic_valuation.html}{\extlink} attached to a height-one prime of a Dedekind domain. Even with this choice, the above discussion about different possible normalizations of a $\wzmZZ$-valued valuations, all leading to the same mathematical object, suggests that the focus should be put on \emph{normalized} valuations, rather than general ones. With this in mind, 
we call a valuation ``discrete'' when it is $\wzmZZ$-valued and \emph{also} normalized, encoding this notion in the following class\href{https://github.com/mariainesdff/local_fields_journal/blob/7ac213eb804fe7945468023527a0fe26ab23b3c8/src/discrete_valuation_ring/basic.lean#L129}{\extlink}:
\begin{lstlisting}[caption={Definition of discrete valuation.}, label={code:def_discrete}]
class is_discrete (v : valuation R $\wzmZZ$) := 
(surj : function.surjective v)
\end{lstlisting}

Observe that when the domain $R$ is endowed with a discrete valuation in the above sense, then it is necessarily a field. Henceforth we change perspective a bit and we focus on a topological field $K$ endowed with a valuation $v$, letting $\ball{K}$ be its unit ball: as discussed at the end of \S\ref{subsection:valuations} this is implemented by putting a \code{valued} instance \code{hv} on \code{K} and by setting 
\begin{lstlisting}
    K$_0$ = hv.v.valuation_subring
\end{lstlisting}

Now, with our definition, a valuation is discrete only if it takes values in the type $\wzmZZ$ and if there exists a uniformizer, and this we prove in the following form\href{https://github.com/mariainesdff/local_fields_journal/blob/7ac213eb804fe7945468023527a0fe26ab23b3c8/src/discrete_valuation_ring/basic.lean#L164}{\extlink}:
\begin{lstlisting}
lemma is_discrete_of_exists_uniformizer {$\pi$ : K} 
  (h$\pi$ : is_uniformizer v $\pi$) : is_discrete v
\end{lstlisting}
The lemma \code{exists_uniformizer_of_discrete}\href{https://github.com/mariainesdff/local_fields_journal/blob/7ac213eb804fe7945468023527a0fe26ab23b3c8/src/discrete_valuation_ring/basic.lean#L324}{\extlink} provides the reverse implication. Similarly, the aforementioned correspondence between uniformizers for a normalized valuation and generators of $\maxid$ takes the form\href{https://github.com/mariainesdff/local_fields_journal/blob/7ac213eb804fe7945468023527a0fe26ab23b3c8/src/discrete_valuation_ring/basic.lean#L295}{\extlink}
\begin{lstlisting}
lemma uniformizer_is_generator 
  ($\pi$ : uniformizer v) :
  maximal_ideal K$_0$ = ideal.span {$\pi.1$}
\end{lstlisting}
and the declaration \code{is_uniformizer_of_generator}\href{https://github.com/mariainesdff/local_fields_journal/blob/7ac213eb804fe7945468023527a0fe26ab23b3c8/src/discrete_valuation_ring/basic.lean#L348}{\extlink} represents the reverse implication. From now on, we call a valued field whose valuation is discrete a \emph{discretely valued field}.

The next result that we want to discuss is the formalization of~\cite[Chapitre~I, \S1, Proposition~1]{Ser62}, stating that the unit ball of a discretely valued field is a DVR. The notion of DVR was already in \mathlib at the time of our project, implemented through the class \code{discrete_valuation_ring}\href{https://leanprover-community.github.io/mathlib_docs/ring_theory/discrete_valuation_ring/basic.html#discrete_valuation_ring}{\extlink} --- this corresponds to property~\ref{dvr_tfae:PID} in Theorem~\ref{thm:dvr_tfae}. Our result takes the following form\href{https://github.com/mariainesdff/local_fields_journal/blob/7ac213eb804fe7945468023527a0fe26ab23b3c8/src/discrete_valuation_ring/basic.lean#L438}{\extlink}:
\begin{lstlisting}[caption={The unit ball is a DVR if the valuation is discrete.}, label={code:serre_prop_1}]
instance dvr_of_is_discrete [is_discrete v] :
  discrete_valuation_ring K$_0$
\end{lstlisting}

Suppose now that $R$ is a DVR, and let $K$ be a field of fractions for $R$. By Theorem~\ref{thm:dvr_tfae}, $R$ is a local ring that is a PID, so in particular $R$ is a Dedekind domain, and one can consider the adic valuation associated to its unique maximal ideal $\maxid$, as defined in Example~\ref{example:adic_valuation}. Now, Baanen~\emph{et al.}~formalized in~\cite{BaaDahNarNuc21} and~\cite{BaaDahNarNuc22} the general theory of Dedekind domains, and de~Frutos-Fernández formalized in~\cite{deF22} the main properties of adic valuations on Dedekind domains.

With these works at our disposal, our starting point is that the $\maxid$-adic valuation $v_{\maxid}$ associated to the maximal ideal $\maxid$ coincides with the normalized valuation $\normalised$ on $K$ (since $R$ is a domain, we directly extend these valuations to any of its fields of fractions). Although this is mathematically trivial, and the two functions $\normalised,v_{\maxid}$ 
are considered identical in pen-and-paper mathematics, they actually belong to different types and hence are different terms: in the \mathlib formalization, the valuation $\normalised$ takes values\href{https://leanprover-community.github.io/mathlib_docs/ring_theory/valuation/valuation_subring.html#valuation_subring.valuation}{\extlink} in the value group $K/R^\times$ (see Remark~\ref{rmk:value_group}), while $v_{\maxid}$ is $\wzmZZ$-valued\href{https://leanprover-community.github.io/mathlib_docs/ring_theory/dedekind_domain/adic_valuation.html#is_dedekind_domain.height_one_spectrum.valuation}{\extlink}. Our approach to compare them is to show that the unit balls with respect to both coincide. To do so, it is enough to show that $R$ is the unit ball of $K$ when endowed with the valuation $v_{\maxid}$, because this is tautologically true with respect to the valuation $\normalised$. To prove the isomorphism $R\cong \ball{K}$ we first provide the following\href{https://github.com/mariainesdff/local_fields_journal/blob/7ac213eb804fe7945468023527a0fe26ab23b3c8/src/discrete_valuation_ring/basic.lean#L456}{\extlink}
\begin{lstlisting}[caption={The \code{valued} instance on the field of frations of a DVR.}, label={code:valued_DVR}]
instance : valued K :=
(maximal_ideal R).adic_valued
\end{lstlisting}
giving a \code{valued} structure on $K$ by using the $\maxid$-adic valuation~$v_{\maxid}$. With this definition, the whole library concerning adic valuations is at our disposal. We then show that the valuation $v_{\maxid}$ is actually discrete (in our sense)  by providing the\href{https://github.com/mariainesdff/local_fields_journal/blob/7ac213eb804fe7945468023527a0fe26ab23b3c8/src/discrete_valuation_ring/basic.lean#L458}{\extlink}
\begin{lstlisting}[caption={The valuation on a DVR is discrete.}, label={code:discrete_on_DVR}]
instance : is_discrete (valued.v K $\wzmZZ$)
\end{lstlisting}
By combining this result with the above discussion we deduce that $\ball{K}$ is itself a DVR, and we finally prove the isomorphism $R\cong\ball{K}$ in the form of the following\href{https://github.com/mariainesdff/local_fields_journal/blob/7ac213eb804fe7945468023527a0fe26ab23b3c8/src/discrete_valuation_ring/basic.lean#L522}{\extlink} 
\begin{lstlisting}
def dvr_equiv_unit_ball : 
  R $\simeq^{ +*}$ valued.v.valuation_subring
\end{lstlisting}

\begin{remark}
When implementing valuations on a discrete valuation ring $R$ with field of fractions\footnote{See~\S\ref{subsec:implementation} for more details.} $K$, we made the design choice to put an instance of the \code{valued} class on $K$, but not on the ring $R$ itself. The reason is that, if needed, we can recover the
uniform structure on $R$ from the one on $K$, so we prefer not to duplicate this information. This choice is consistent with \mathlib's file \code{ring\_theory.dedekind\_domain.adic\_valuation}\href{https://leanprover-community.github.io/mathlib_docs/ring_theory/dedekind_domain/adic_valuation.html}{\extlink}, in which \code{valued} instances are only put on fields.
\end{remark}

As illustrated in the Introduction, one application of the theory of DVR's to number theory comes from the fact that every maximal ideal~$\maxid[\empty]$ in a Dedekind domain $R$ induces a discrete valuation, with respect to which the unit ball is a DVR. In this context, we prove the following lemma:\href{https://github.com/mariainesdff/local_fields_journal/blob/7ac213eb804fe7945468023527a0fe26ab23b3c8/src/discrete_valuation_ring/global_to_local.lean#L28}{\extlink}
\begin{lstlisting}
lemma adic_valued_is_discrete
  [is_dedekind_domain R] [is_fraction_ring R K]
  ($\maxid[\empty]$ : height_one_spectrum R) :
  is_discrete (adic_valued $\maxid[\empty]$).v
\end{lstlisting}

\subsection{Complete Discrete Valuation Rings}\label{subsec:complete_DVR}
One setting where the theory of DVR's becomes crucial --- especially for its applications to algebraic number theory and algebraic geometry --- is that of (adic) completions. The completion is a general procedure that can be performed on every uniform space $T$, as explained in~\cite[Chapitre~II, \S3]{Bou71}. We are not providing the exact definition here; the crucial property to retain is that it yields another uniform space $\compl{T}$, containing $T$ as a dense subspace, and such that every uniformly continuous function $f\colon T \to Z$, valued in any complete, separated topological space $Z$ can be extended \emph{uniquely} to a function
\begin{equation}\label{eq:univ_compl}
\compl{f}\colon\compl{T}\rightarrow Z
\end{equation}
(see~\cite[Chapitre~II, \S3, n$^\circ$6, Théorème~2]{Bou71}). When $T$ is a commutative topological group, it suffices that $f$ is a continuous group homomorphism to admit a unique extension, because \cite[Chapitre~III, \S3, n$^\circ$1, Proposition~3]{Bou71} guarantees that it is automatically uniformly continuous; and the completeness on $Z$ can be weakened as in~\cite[Chapitre~I, \S8, n$^\circ$5, Théorème~1]{Bou71}.

In the special case where $T=K$ is a field with a valuation~$v$ (endowed with the uniform structure induced by it as explained in \S\ref{subsection:valuations}), the completion $\compl{K}$ is again a field. We refer to~\cite[Chapitre~VI, \S5, n$^\circ$3]{Bou85} for the generalities concerning completions of valued fields, and to~\cite[Chapitre~II, \S1]{Ser62} for a shorter introduction. Applying\footnote{We do not describe here the topology on $\wzmZZ$, suffices it to say that it mimics the discrete one on $\ZZ$. We refer the reader to the relevant \mathlib file\href{https://leanprover-community.github.io/mathlib_docs/topology/algebra/with_zero_topology.html}{\extlink}.} the universal property~\eqref{eq:univ_compl} to the valuation $v$ yields a map $\compl{v}\colon \compl{K}\rightarrow \wzmZZ$ that is still a valuation. Completions of uniform spaces and uniform fields have been formalized\footnote{Given a ring $A$, an ideal $I\subseteq A$ and an $A$-module $M$, \mathlib contains the declaration \code{adic_completion I M}. This is defined purely algebraically as a module of ``coherent sequences'', a concrete incarnation of an inverse limit. A systematic comparison between this $I$-adic completion and the uniform one has not yet been formalized.}, and the paper~\cite{BuzComMas20} contains an account of the formalization. The application of this construction in the setting of Dedekind domains is due to de~Frutos-Fernández, see~\cite{deF22}. 

A prototypical example of a completion of a field is the field $\QQ_p$ of $p$-adic numbers, that is defined as the completion of~$\QQ$ endowed with the $p$-adic valuation $v_p$. Another example arises when starting with the ring $K[X]$ of polynomials in one indeterminate over the field $K$: being a PID, it is a Dedekind domain and the construction of Example~\ref{example:adic_valuation} defines a valuation $v\colon K(X)\to\wzmZZ$ attached to any maximal ideal $\maxid[\empty]\subseteq K[X]$. Take, for instance, $\maxid[\empty]=(X)$. Then the valuation $v_X$ defines a uniform structure on the field
\[
K(X)=\Frac(K[X])=\Bigl\{\frac{f}{g},\; f,g\in K[X]\text{ with }g\ne 0\Bigr\}
\]
of rational functions whose completion $\FpXCompl[K]$ is isomorphic to the field $\laurentseries[K]$ of Laurent series: we discuss our formalization of this isomorphism in~\S\ref{subsubsection:isom_laurent}.

When completing a field $K$ endowed with a valuation $v$, the valuation $\compl{v}$ is unique under the condition of extending the one on $K$, by~\eqref{eq:univ_compl}. It follows that $\compl{K}$ is always endowed with a global \code{valued} instance, even when $K$ is not. This is for example the case with $\QQ$ and $\QQ_p$: the first has infinitely many non-equivalent valuations $v_p$, and thus no global \code{valued} instance is declared for it. On the other hand, each $\QQ_p$ has a preferred valuation and is actually a valued field. More generally, given a maximal ideal $\maxid[\empty]$ of a Dedekind domain $R$, we can apply the above discussion to the valuation $v_{\maxid[\empty]}$ to obtain a complete valued field $\compl[{\maxid[\empty]}]{K}$\href{https://leanprover-community.github.io/mathlib_docs/ring_theory/dedekind_domain/adic_valuation.html#is_dedekind_domain.height_one_spectrum.valued_adic_completion}{\extlink}.

We now go back to our discussion concerning DVR's. So, let $R$ be a DVR with field of fractions $K$. We have seen in the \LClistingname~\ref{code:valued_DVR} that $K$, endowed with the valuation $v_{\maxid}$, is valued and so is its completion $\compl[\maxid]{K}$, denoted by \code{K_v} in our code. Moreover, the valuation on $K$ is discrete (see the \LClistingname~\ref{code:discrete_on_DVR}) and the first result we obtain is that the same holds for the valuation $\compl{v_{\maxid}}$ on $\compl[\maxid]{K}$, which takes the following form\href{https://github.com/mariainesdff/local_fields_journal/blob/7ac213eb804fe7945468023527a0fe26ab23b3c8/src/discrete_valuation_ring/complete.lean#L63}{\extlink}:
\begin{lstlisting}
instance (R : Type*) [is_dedekind_domain R]
  (v : height_one_spectrum R) : 
  is_discrete (valued.v K_v $\wzmZZ$)
\end{lstlisting}

In particular, we can consider the ring $\ball{\bigl(\compl[\maxid]{K}\bigr)}$ (denoted by \code{R_v} in our code) to find, by the discussion in~\S\ref{subsection:def_dvr}, that it is itself a DVR\href{https://github.com/mariainesdff/local_fields_journal/blob/7ac213eb804fe7945468023527a0fe26ab23b3c8/src/discrete_valuation_ring/complete.lean#L77}{\extlink}
\begin{lstlisting}
instance : discrete_valuation_ring R_v :=
discrete_valuation.dvr_of_is_discrete _
\end{lstlisting}
with field of fractions $\compl[\maxid]{K}$. Thus, $\ball{\bigl(\compl[\maxid]{K}\bigr)}$ is a local ring endowed with a maximal ideal, denoted $\compl{\maxid}$. It is therefore possible to once again apply de~Frutos-Fernández' work and endow $\compl[\maxid]{K}=\Frac\bigl(\ball{\bigl(\compl[\maxid]{K}\bigr)}\bigr)$ with the $\compl{\maxid}$-adic valuation $v_{\compl{\maxid}}$, which puts a (potentially) new \code{valued} structure on $\compl[\maxid]{K}$.

Now, as the notation suggests, the maximal ideal $\compl{\maxid}$, which is in particular an abelian group, actually coincides with the completion (inside the larger space $\compl[\maxid]{K}$) of the abelian group $\maxid$. In our language, this reflects on the equality of the two valuations $\compl{v_{\maxid}}$ and $v_{\compl{\maxid}}$ on $\compl[\maxid]{K}$, but this equality will clearly not be a definitional one, given the different constructions that led to the two valuations. Rather, it takes the following form:\href{https://github.com/mariainesdff/local_fields_journal/blob/7ac213eb804fe7945468023527a0fe26ab23b3c8/src/discrete_valuation_ring/complete.lean#L184}{\extlink}
\begin{lstlisting}
lemma adic_of_compl_eq_compl_of_adic (x : K_v) :
  v_adic_of_compl x = v_compl_of_adic x
\end{lstlisting}
In the above \LClistingname, \code{v_adic_of_compl}\href{https://github.com/mariainesdff/local_fields_journal/blob/7ac213eb804fe7945468023527a0fe26ab23b3c8/src/discrete_valuation_ring/complete.lean#L137}{\extlink} represents the valuation $v_{\compl{\maxid}}$, while \code{v_compl_of_adic}\href{https://github.com/mariainesdff/local_fields_journal/blob/7ac213eb804fe7945468023527a0fe26ab23b3c8/src/discrete_valuation_ring/complete.lean#L141}{\extlink} is $\compl{v_{\maxid}}$.

\subsection{Extensions of Complete Discrete Valuation Rings}\label{subsec:extensions}
The goal of this section is to prove that if $K$ is any field complete with respect to a discrete valuation $v$ and $L/K$ is a finite extension of fields, then there is a unique discrete valuation on $L$ ``extending'' $v$. To explain our formalization, we first need to discuss the relation between valuations and norms.

\begin{definition}\label{def:nonarchimedean_norm}
A \emph{nonarchimedean multiplicative norm} on a ring $R$ is a function $\lvert\cdot\rvert \colon R \to \RR$
such that
\begin{listDef}
\item $\lvert r\rvert = 0$ if and only if $r = 0$ for all $r \in R$;
\item $\lvert 1\rvert = 1$;
\item $\lvert r\cdot s\rvert  = \lvert r\rvert \cdot \lvert s\rvert $ for all $r, s$ in $R$;
\item $\lvert r + s\rvert \leq \max \{\lvert r\rvert, \lvert s\rvert \}$ for all $r, s$ in $R$;
\item $\lvert -r\rvert = \lvert r\rvert$ for all $r$ in $R$.
\end{listDef}
It follows from these conditions that $0 \le \lvert r\rvert$ for all $r \in R$.
\end{definition}
A valuation $v \colon R \to \Gamma_0$ valued in a group with zero $\Gamma_0$ has \emph{rank one} if it is nontrivial and there exists an injective morphism of linearly ordered groups with zero from $\Gamma_0$ to $\RR_{\ge 0}$. In particular, any nontrivial valuation $v \colon R \to \wzmZZ$ has rank one. The definition of rank one valuation was first formalized in \cite{deF23}\href{https://github.com/mariainesdff/norm_extensions_journal_submission/blob/d396130660935464fbc683f9aaf37fff8a890baa/src/rank_one_valuation.lean#L38}{\extlink}.

The conditions in Definitions~\ref{def:valuation} and~\ref{def:nonarchimedean_norm} are analogous --- apart from the codomain of the function --- and the terms ``multiplicative valuation'' and ``nonarchimedean multiplicative norm'' are often used interchangeably in the mathematical literature. While \mathlib does not yet provide a way to relate these two notions, a dictionary between them has been formalized in \cite{deF23}. Namely, if $K$ is a field with a nonarchimedean norm, the definition \code{valuation_from_norm}\href{https://github.com/mariainesdff/norm_extensions_journal_submission/blob/d396130660935464fbc683f9aaf37fff8a890baa/src/normed_valued.lean#L37}{\extlink} yields the corresponding valuation $v \colon K \rightarrow \RR_{\ge 0}$. Conversely, if $L$ is a field with a rank one valuation, then \code{norm_def}\href{https://github.com/mariainesdff/norm_extensions_journal_submission/blob/d396130660935464fbc683f9aaf37fff8a890baa/src/normed_valued.lean#L110}{\extlink} is the corresponding norm function on $L$.
To  make use of this dictionary in our project, we need to provide an injective morphism $\wzmZZ \to \RR_{\ge 0}$ of linearly ordered groups with zero. We can define this morphism by picking any real $n > 1$ and identifying $\wzmZZ$ with the subset $n^\ZZ \cup \{0\}$ of $\RR_{\ge 0}$, via the map sending $0$ to $0$ and \code{of_add(x)} to $n^x$, for $x \in \ZZ$.

We construct this morphism as\href{https://github.com/mariainesdff/local_fields_journal/blob/7ac213eb804fe7945468023527a0fe26ab23b3c8/src/for_mathlib/with_zero.lean#L157}{\extlink} 
\begin{lstlisting}
def with_zero_mult_int_to_nnreal {n : $\RR_{\ge 0}$}
  (he : n $\ne 0$)  : $\wzmZZ \to_{*0} \; \RR_{\ge 0}$ := 
{ to_fun := $\lambda$ x, if hx : n = 0 then 0 else 
    n^(to_add (with_zero.unzero hx)),
  ... }
\end{lstlisting}
and then we prove\href{https://github.com/mariainesdff/local_fields_journal/blob/7ac213eb804fe7945468023527a0fe26ab23b3c8/src/for_mathlib/with_zero.lean#L191}{\extlink} that this map is order-preserving whenever $n > 1$. Each choice of $n$ gives rise to a different morphism, and hence to a different norm attached to a valuation $v \colon K \to \wzmZZ$, although any two such norms define the same uniformity on $K$. When the quotient $\ball{K}/\maxid[\ball{K}]$ is finite of order $p^k$ for some prime $p$, the standard choice is $n = p^k$. Otherwise there is not a preferred $n$, so in our formalization we pick $n=6$. We refer to this $n$ associated to $v$ as \code{v.base}\href{https://github.com/mariainesdff/local_fields_journal/blob/7ac213eb804fe7945468023527a0fe26ab23b3c8/src/discrete_valuation_ring/basic.lean#L208}{\extlink}.

Let now $K$ be a complete discretely valued field and let $L$ be a finite field extension, so that $L/K$ is automatically algebraic. The above discussion allows us to apply the following theorem to conclude that there is a unique norm $\lvert\cdot\rvert_L$ on $L$ extending the norm $\lvert \cdot\rvert _K$ associated with the valuation on $K$.

\begin{theorem}\label{th:uniq}
Let $K$ be a field that is complete with respect to a nonarchimedean multiplicative norm $\lvert\cdot\rvert _K$ and let $L/K$ be an algebraic extension. Then     there is a unique multiplicative nonarchimedean norm on $L$, called the \emph{spectral norm}, extending the norm $\lvert\cdot\rvert_K$.
\end{theorem}

\begin{proof}
See \cite[3.2.4/2]{BGR} for the informal proof and \cite[\S3.2]{deF23} for a discussion of the \lean[\empty] formalization.
\end{proof}
In the setting that we are considering, the norm whose existence is guaranteed by Theorem~\ref{th:uniq} is formalized in the declaration \code{discrete_norm_extension}\href{https://github.com/mariainesdff/local_fields_journal/blob/7ac213eb804fe7945468023527a0fe26ab23b3c8/src/discrete_valuation_ring/discrete_norm.lean#L123}{\extlink}.

Since the extension $L/K$ is algebraic, there is an explicit formula for the norm $\lvert x\rvert_L$ of an element $x \in L$: if $x$ has minimal polynomial $f_x(X) = X^m + a_{m-1}X^{m-1} + \dots + a_0$, then $\lvert x\rvert _L = \lvert a_0\rvert _K^{1/m}$. We prove this in\href{https://github.com/mariainesdff/local_fields_journal/blob/7ac213eb804fe7945468023527a0fe26ab23b3c8/src/spectral_norm.lean#L171}{\extlink}:

\begin{lstlisting}
theorem spectral_norm_eq_root_zero_coeff :
  spectral_norm K L x = $\|$(minpoly K x).coeff 0$\|$ ^ (1/(minpoly K x).nat_degree : $\RR$)
\end{lstlisting}

Since for every $x \in L$, the degree of the minimal polynomial $f_x$ of $x$ over $K$ divides the degree $[L : K]$ of the extension $L/K$\href{https://github.com/mariainesdff/local_fields_journal/blob/7ac213eb804fe7945468023527a0fe26ab23b3c8/src/for_mathlib/field_theory/minpoly/is_integrally_closed.lean#L31}{\extlink}, we see that the norm $\lvert\cdot\rvert_L$ takes values in the subset $n^{\frac{\ZZ}{[L : K]}} \cup \{0\}$ of $\RR_{\ge 0}$. The norm needs not surject onto that subset but, since $\ZZ$ is cyclic, its image is of the form $n^{\frac{b\ZZ}{[L : K]}} \cup \{0\}$ for some $b \in \ZZ_{>0}$. Informally speaking, to obtain the corresponding normalized valuation we just need to rescale the norm by raising it to the $([L:K]/b)$-th power, and use the identification between $n^\ZZ \cup \{0\}$ and $\wzmZZ$.

In order to formalize this idea, we have to be careful to proceed in such a way that we do not leave the type $\wzmZZ$, since taking the $b^{\text{th}}$-root of an element of $\wzmZZ$ is not a well-defined operation. 
We start by constructing a map\href{https://github.com/mariainesdff/local_fields_journal/blob/7ac213eb804fe7945468023527a0fe26ab23b3c8/src/discrete_valuation_ring/extensions.lean#L107}{\extlink}
\begin{lstlisting}
 pow_extension_on_units : L$^\times\to$ multiplicative $\ZZ$
\end{lstlisting}
that sends $x \in L^\times$ to $v(a_0)^\frac{[L : K]}{\deg f_x}$, which is well-defined by lemma \code{unit_pow_ne_zero}\href{https://github.com/mariainesdff/local_fields_journal/blob/7ac213eb804fe7945468023527a0fe26ab23b3c8/src/for_mathlib/ring_theory/valuation/minpoly.lean#L45}{\extlink}. Then we find the positive number \code{exp_extension_on_units}\href{https://github.com/mariainesdff/local_fields_journal/blob/7ac213eb804fe7945468023527a0fe26ab23b3c8/src/discrete_valuation_ring/extensions.lean#L127}{\extlink}, denoted $b$ above, such that the image of \code{pow_extension_on_units} is generated by \code{of_add (exp_extension_on_units} $ \colon \ZZ$ \code{)}. Hence, using the lemma \code{exists_mul_exp_extension_on_units}\href{https://github.com/mariainesdff/local_fields_journal/blob/7ac213eb804fe7945468023527a0fe26ab23b3c8/src/discrete_valuation_ring/extensions.lean#L193}{\extlink}, we can obtain the natural number $c$ such that 
\begin{lstlisting}
pow_extension_on_units K L x = 
  (of_add (exp_extension_on_units))^c;
\end{lstlisting} which leads to our definition of the extended valuation\href{https://github.com/mariainesdff/local_fields_journal/blob/7ac213eb804fe7945468023527a0fe26ab23b3c8/src/discrete_valuation_ring/extensions.lean#L246}{\extlink}:
\begin{lstlisting}
def extension_def : L $\to \wzmZZ$ := $\lambda$ x, 
  if hx : x = 0 then 0 else (of_add (-1 : $\ZZ$)) ^ 
    (exists_mul_exp_extension_on_units K 
      (is_unit_iff_ne_zero.mpr hx).unit).some
\end{lstlisting}

To connect the extended norm and the extended valuation we prove in \code{pow_eq_pow_root_zero_coeff}\href{https://github.com/mariainesdff/local_fields_journal/blob/7ac213eb804fe7945468023527a0fe26ab23b3c8/src/discrete_valuation_ring/discrete_norm.lean#L158}{\extlink} that for every multiple $d$ of $\deg f_x$, the following equality holds:
\begin{equation}\label{eq:pow_eq_pow_root}
\mcode{of\_add}
\Bigl(\log_n\bigl(\lvert x\rvert_L^{d}\bigr)\Bigr)=v(a_0)^\frac{d}{\deg f_x}
\end{equation}
Using~\eqref{eq:pow_eq_pow_root} we deduce that \code{extension_def K L} satisfies the properties in Definition~\ref{def:valuation} since \code{spectral_norm K L} is a multiplicative nonarchimedean norm on $L$ (as proven in \cite[\S3.2]{deF23}). We then prove that this valuation is discrete\href{https://github.com/mariainesdff/local_fields_journal/blob/7ac213eb804fe7945468023527a0fe26ab23b3c8/src/discrete_valuation_ring/extensions.lean#L447}{\extlink} and that  $L$ is complete with respect to the induced uniform structure\href{https://github.com/mariainesdff/local_fields_journal/blob/7ac213eb804fe7945468023527a0fe26ab23b3c8/src/discrete_valuation_ring/extensions.lean#L561}{\extlink}. 

The next theorem we prove says that the integral closure $\rint{L}$ of the unit ball $\ball{K}$ inside $L$ is again a discrete valuation ring\href{https://github.com/mariainesdff/local_fields_journal/blob/7ac213eb804fe7945468023527a0fe26ab23b3c8/src/discrete_valuation_ring/extensions.lean#L616}{\extlink} (we refer to~\cite[Chapitre~5, \S2, n$^\circ$1]{Bou85} for generalities about integral closures). This follows from the fact that $\rint{L}$ is actually equal to the unit ball with respect to the extended discrete valuation on $L$, formalized as\href{https://github.com/mariainesdff/local_fields_journal/blob/7ac213eb804fe7945468023527a0fe26ab23b3c8/src/discrete_valuation_ring/extensions.lean#L588}{\extlink}
\begin{lstlisting}
lemma integral_closure_eq_integer :
  integral_closure hv.v.valuation_subring L =
  (extended_valuation K L).valuation_subring
\end{lstlisting}

\section{Local Fields}\label{sec:local_fields}
In this section we describe our formalization of the basic theory of local fields, a special kind of discretely valued fields which are fundamental objects of study in number theory. For instance, given a number field $L$ with ring of integers $\rint{L}$, one can subsequently apply the constructions of Example~\ref{example:adic_valuation} to the different maximal ideals $\maxid[\empty]\subseteq\rint{L}$ --- thus obtaining a collection of DVR's --- followed by the completion procedure described in \S\ref{subsec:complete_DVR}. The resulting collection of fields $\compl[{\maxid[\empty]}]{L}$ are precisely the local fields occurring in the ``local-to-global'' approach to class field theory briefly mentioned in the Introduction.

Before giving the definition, we make a preliminary observation. Theorem~\ref{thm:dvr_tfae} shows that every DVR is a local ring, so it has a unique maximal ideal. On the other hand, given a field $K$ endowed with a discrete valuation, the \LClistingname~\ref{code:serre_prop_1} shows that its unit ball $\ball{K}$ is a DVR, whose maximal ideal we denote $\maxid[\ball{K}]$. It is then possible to unambiguously (although slightly inappropriately) speak of the \emph{residue field} of $K$ to mean the field $\ball{K}/\maxid[\ball{K}]$.

\begin{definition}\label{def:local_field}
A \emph{local field} is a field complete with respect to a discrete valuation and with finite residue field.
\end{definition}
The definition is implemented as follows:\href{https://github.com/mariainesdff/local_fields_journal/blob/0b408ff3af36e18f991f9d4cb87be3603cfc3fc3/src/local_field.lean#L35}{\extlink}
\begin{lstlisting}
class local_field (K : Type*) [field K] 
  [hv : valued K $\wzmZZ$] :=
(complete : complete_space K)
(is_discrete : is_discrete hv.v)
(finite_residue_field : fintype (local_ring.residue_field hv.v.valuation_subring))
\end{lstlisting}

\begin{remark}
What we call ``local fields'' are normally referred to as ``nonarchimedean local fields'', to distinguish them from the archimedean local fields $\RR$ and $\CC$. Since we only consider the nonarchimedean case in this work, we have opted for this simplification. Observe that the requirement that the residue field is finite is sometimes omitted. Accordingly, the completion $\FpXCompl[\CC]$ is a valued field with residue field $\CC$ that we do not qualify to be local, although this is sometimes the case in the literature.
\end{remark}

A local field $K$ can be of two kinds: of \emph{equal characteristic}, if both $K$ and its residue field have characteristic $p$ for some prime number $p$, or of \emph{mixed characteristic}, if $K$ has characteristic $0$ and its residue field has characteristic $p$. Moreover, in the first case $K$ is a finite extension of $\laurentseries$ while in the second it is a finite extension of $\QQ_p$ (see~\cite[Chapitre~VI,~\S9, n$^\circ$3]{Bou85}); this motivates our definitions in~\S\S\ref{subsec:eq_char}--\ref{subsec:mixed_char}.

\subsection{Equal characteristic}\label{subsec:eq_char}
Recall that we denote by $\FpXCompl$ the completion of $\FF(X)$ with respect to the maximal ideal $\maxid[\empty]=(X)$.
\begin{definition}\label{def:eq_char_local_field}
An \emph{equal characteristic local field} is a finite dimensional field extension of \FpXCompl, for some prime number $p$.
\end{definition}
This definition is formalized as\href{https://github.com/mariainesdff/local_fields_journal/blob/0b408ff3af36e18f991f9d4cb87be3603cfc3fc3/src/eq_characteristic/basic.lean#L359}{\extlink}:
\begin{lstlisting}
class eq_char_local_field (p : $\NN$) [nat.prime p]
  (K : Type*) [field K] 
  extends algebra (FpX_completion p) K :=
[to_finite_dimensional : finite_dimensional 
  (FpX_completion p) K]
\end{lstlisting}

The first example of equal characteristic local field is $\FpXCompl$ itself, and we record our proof of this in the instance \code{FpX_completion.eq_char_local_field}\href{https://github.com/mariainesdff/local_fields_journal/blob/0b408ff3af36e18f991f9d4cb87be3603cfc3fc3/src/eq_characteristic/basic.lean#L440}{\extlink}. Being a finite extension of the complete discretely valued field $\FpXCompl$, any equal characteristic  local field $K$ is endowed with a unique nontrivial valuation, which is again discrete, and $K$ is complete with respect to it.
\begin{lstlisting}
instance : valued K $\wzmZZ$ := 
extension.valued (FpX_completion p) K
instance : complete_space K := 
extension.complete_space (FpX_completion p) K
instance : is_discrete 
  (eq_char_local_field.valued p K).v := 
extension.is_discrete_of_finite 
  (FpX_completion p) K
\end{lstlisting}

\begin{remark}\label{rmk:laurent_vs_adic}
Given any field $K$, the $X$-adic completion $\FpXCompl[K]$ of the field of rational functions is isomorphic to the field of Laurent series indexed by $\ZZ$ with only finitely many negative nonzero coefficients:
\[
\laurentseries[K]=\Bigl\{f=\sum_{n \in\ZZ}a_nX^n \mid a_n \in K\text{ and }a_n = 0 \text{ if }n\ll 0\Bigr\}.
\]
Accordingly, the unit ball $\FpXComplInt[K]$ is isomorphic to the ring of power series
\[
\powerseries[K]=\Bigl\{\sum a_nX^n \in \laurentseries[K]\mid a_n=0 \,\forall\,n\le 0, \Bigr\}.
\]

It is customary to define an equal characteristic local field as a finite extension of $\laurentseries$ rather than of $\FpXCompl$, because elements in $\laurentseries$ are explicit and it might be handier to work with them than with elements in $\FpXCompl$. While in pen-and-paper mathematics one can safely treat finite extensions of the two fields as leading to the same definition, we have to pick a choice in our formalization project. Since our development of extensions of complete DVR is very general and does not rely on any explicit description of the base field, and in analogy with our approach to the mixed characteristic case (see~\S\ref{subsec:mixed_char}), we decide to use the type $\FpXCompl$ as the base field in our Definition~\ref{def:mixed_char_local_field}. Moreover, the isomorphism $\laurentseries[K]\cong \FpXCompl[K]$ is still not available in \mathlib and we describe our formalization of this isomorphism in~\S\ref{subsubsection:isom_laurent}.
\end{remark}

Our next task is to show that an equal characteristic local field $K$ is a local field, in the sense of Definition~\ref{def:local_field}. We have formalized this proof under the hypothesis that the extension $K/\FpXCompl$ is separable, as this is required to apply the \mathlib lemma \code{is_integral_closure.is_noetherian}\href{https://leanprover-community.github.io/mathlib_docs/ring_theory/dedekind_domain/integral_closure.html#is_integral_closure.is_noetherian}{\extlink} which we are using in our proof of the finiteness of the residue field. However, we point out that the separability assumption can be removed at the expense of a more involved proof, which we plan to formalize at a later date.

Since $\FpXCompl$ is a field complete with respect to a discrete valuation and $K$ is a finite extension of $\FpXCompl$, by the discussion following~\eqref{eq:pow_eq_pow_root}, the $X$-adic valuation on $\FpXCompl$ induces a complete\href{https://github.com/mariainesdff/local_fields_journal/blob/0b408ff3af36e18f991f9d4cb87be3603cfc3fc3/src/eq_characteristic/valuation.lean#L57}{\extlink} discretely\href{https://github.com/mariainesdff/local_fields_journal/blob/0b408ff3af36e18f991f9d4cb87be3603cfc3fc3/src/eq_characteristic/valuation.lean#L60}{\extlink} \code{valued} structure on $K$, registered in the instance \code{eq_char_local_field.valued p K}\href{https://github.com/mariainesdff/local_fields_journal/blob/0b408ff3af36e18f991f9d4cb87be3603cfc3fc3/src/eq_characteristic/valuation.lean#L54}{\extlink}. Hence it only remains to prove that the residue field of $K$ is finite. To prove the finiteness statement, we first show that if $E$ is a field complete with respect to a discrete valuation and $L/E$ is a finite separable field extension, then the residue field of $L$ is finite dimensional over the residue field of $E$\href{https://github.com/mariainesdff/local_fields_journal/blob/0b408ff3af36e18f991f9d4cb87be3603cfc3fc3/src/discrete_valuation_ring/residue_field.lean#L379}{\extlink}:
\begin{lstlisting}
  finite_dimensional (residue_field E$_0$) 
    (residue_field (integral_closure E$_0$ L))
\end{lstlisting}
This implies that if the residue field of $E$ is finite, then so is the residue field of $L$. Now, it follows from the isomorphism $\laurentseries\cong\FpXCompl$ discussed in~\S\ref{subsubsection:isom_laurent} that $E=\FpXCompl$ has residue field $\FF$, so every equal characteristic local field has finite residue field, and is therefore a local field.

The \emph{ring of integers} of an equal characteristic local field $K$, denoted $\rint{K}$, is the integral closure of $\FpXComplInt$ in $K$\href{https://github.com/mariainesdff/local_fields_journal/blob/0b408ff3af36e18f991f9d4cb87be3603cfc3fc3/src/eq_characteristic/basic.lean#L375}{\extlink}:
\begin{lstlisting}
def ring_of_integers := 
integral_closure (FpX_int_completion p) K
\end{lstlisting}
The lemma \code{integral_closure_eq_integer} implies that $\rint{K} = \ball{K}$, that is, an element of $K$ is integral over $\FpXComplInt$ if and only if its valuation is less than or equal to $1$. A pivotal consequence of this equality is that $\rint{K}$ is a discrete valuation ring\href{https://github.com/mariainesdff/local_fields_journal/blob/0b408ff3af36e18f991f9d4cb87be3603cfc3fc3/src/eq_characteristic/valuation.lean#L64}{\extlink}:

\begin{lstlisting}
instance : discrete_valuation_ring ($\rint{}$ p K) := 
integral_closure.dvr_of_finite_extension 
  (FpX_completion p) K
\end{lstlisting}

\subsubsection{Formalizing the isomorphism $\laurentseries[K]\cong\FpXCompl[K]$}\label{subsubsection:isom_laurent}

Due to the relatively recent appearance of the theory of adic valuations in \mathlib, the isomorphism $\laurentseries[K]\cong\FpXCompl[K]$ was not formalized at the time of our work. Nevertheless the API for working with completions of uniform spaces and uniform fields is quite rich, as described in~\S\ref{subsec:complete_DVR}. The key ingredient for the formalization is the notion of \emph{abstract completion}\href{https://leanprover-community.github.io/mathlib_docs/topology/uniform_space/abstract_completion.html#abstract_completion}{\extlink} of uniform spaces. Given a uniform space~$T$, a term \code{pkg : abstract_completion T} contains seven fields, the three most relevant for us being
\begin{listDef}
\item \code{pkg.space}, the underlying uniform space endowed with a map \code{coe : pkg} $\to$ \code{T};
\item \code{pkg.complete}, which is a proof that \code{pkg.space} is complete;
\item \code{pkg.dense_coe}, which is a proof that \code{coe} is injective with dense image.
\end{listDef}
In particular, there is a term \code{ratfunc_adic_compl_pkg}\href{https://github.com/mariainesdff/local_fields_journal/blob/0b408ff3af36e18f991f9d4cb87be3603cfc3fc3/src/laurent_series_equiv_adic_completion.lean#L785}{\extlink} whose \code{space} field represents the completion $\FpXCompl$. Now, given two terms \code{(pkg, pkg': abstract_completion T)}, the declaration \code{compare pkg pkg'}\href{https://leanprover-community.github.io/mathlib_docs/topology/uniform_space/abstract_completion.html#abstract_completion.compare}{\extlink} provides an equivalence of uniform spaces between them, which expresses the mathematical statement that ``the completion of a uniform space is unique up to a unique isomorphism''. The uniqueness of the extension~\eqref{eq:univ_compl} makes it easy to upgrade the equivalence to an equivalence of fields whenever both \code{pkg.space} and \code{pkg'.space} are fields. It follows that once we prove that $\laurentseries[K]$ is also an abstract completion of $K(X)$, it will correspond to a term \code{laurent_series_pkg}\href{https://github.com/mariainesdff/local_fields_journal/blob/0b408ff3af36e18f991f9d4cb87be3603cfc3fc3/src/laurent_series_equiv_adic_completion.lean#L790}{\extlink} and the previous discussion produces the required isomorphism, in the form\href{https://github.com/mariainesdff/local_fields_journal/blob/0b408ff3af36e18f991f9d4cb87be3603cfc3fc3/src/laurent_series_equiv_adic_completion.lean#L841}{\extlink}
\begin{lstlisting}[caption={The isomorphim between the Laurent series and the completion of rational functions.}, label={code:isom_ls}]
def laurent_series_ring_equiv :
   (laurent_series K) $\simeq ^{+*}$ (ratfunc_adic_compl K)
\end{lstlisting}

To prove that $\laurentseries[K]$ is a completion of $K(X)$, we need to show that $\laurentseries[K]$ is complete and that the image of the coercion $K(X)\hookrightarrow \laurentseries[K]$ is dense. Ultimately, both results rely on a careful study of the interaction between the $X$-adic valuation on $K(X)$, the $X$-adic valuation on $\laurentseries[K]$\footnote{Despite having the same name, these valuations are associated to different ideals: one is $(X)\subseteq K[X]$ and the other is $(X)\subseteq \powerseries[K]$.}, and the coefficients of the corresponding series. For example, we show in the lemma\href{https://github.com/mariainesdff/local_fields_journal/blob/0b408ff3af36e18f991f9d4cb87be3603cfc3fc3/src/laurent_series_equiv_adic_completion.lean#L372}{\extlink}
\begin{lstlisting}
lemma valuation_le_iff_coeff_zero_of_lt {D : $\ZZ$}
  {f : laurent_series K} : v f $\le$ (of_add (-D))
  $\leftrightarrow$ ($\forall$ n : $\ZZ$, n < D $\to$ f.coeff n = 0)
\end{lstlisting}
that a Laurent series has valuation bounded by \code{of_add (-D)} if and only if its $n$-th coefficient vanishes for each $n<D$. Once we can relate the valuation and the vanishing of the coefficients, the proof that $K(X)$ has dense image in $\laurentseries[K]$ is very smooth. The proof of completeness heavily relies on the formalism of filters, as explained in~\cite[\S4]{BuzComMas20}. The main ingredient is the\href{https://github.com/mariainesdff/local_fields_journal/blob/0b408ff3af36e18f991f9d4cb87be3603cfc3fc3/src/laurent_series_equiv_adic_completion.lean#L490}{\extlink}
\begin{lstlisting}
lemma uniform_continuous_coeff (d : $\ZZ$)
  (h : uniformity K = $\princfilter$ id_rel) :
  uniform_continuous (coeff d)
\end{lstlisting}
It shows that if $K$ is endowed with the \emph{discrete} uniformity (\cite[Chapitre~II, \S1, n$^\circ$1, Exemple~2]{Bou71}), the map $f \mapsto a_d(f)$ sending a Laurent series to its $d$-th coefficient is uniformly continuous, for every $d\in\ZZ$. The consequence (see \cite[Chapitre~II, \S3, n$^\circ$1, Proposition~3]{Bou71}) is that for every Cauchy filter $\filter$ in $\laurentseries[K]$, the push-forward $f(\filter)$ is a Cauchy filter of the discrete space $K$ and thus converges to a point $c_d(\filter)$. It is then easy to combine the above results linking coefficients and valuation to show that, for every Cauchy filter $\filter$, the value $c_d(\filter)$ vanishes for $d\ll 0$ and therefore
\[
f(\filter)=\sum_{d\in\ZZ} c_dX^d\in\laurentseries[K].
\]
In the lemma \code{cauchy.eventually_mem_nhds}\href{https://github.com/mariainesdff/local_fields_journal/blob/0b408ff3af36e18f991f9d4cb87be3603cfc3fc3/src/laurent_series_equiv_adic_completion.lean#L630}{\extlink} we then show that $\filter$ converges to the principal filter $\princfilter(f(\filter))$, proving that $\laurentseries[K]$ is complete. Finally, we can specialize to $\powerseries[K]$ the equivalence \code{laurent_series_ring_equiv}
to get its integral version:\href{https://github.com/mariainesdff/local_fields_journal/blob/0b408ff3af36e18f991f9d4cb87be3603cfc3fc3/src/laurent_series_equiv_adic_completion.lean#L983}{\extlink}
\begin{lstlisting}[caption={The isomorphism between power series and the unit ball in the completion of rational functions.}, label={code:iso_ps}]
def power_series_ring_equiv : (power_series K) 
  $\simeq ^{+*}$ ((ideal_X K).adic_completion_integers 
        (ratfunc K))
\end{lstlisting}
In particular, we see that the residue field of $\FpXCompl[K]$ is isomorphic to the residue field of $\laurentseries[K]$, hence to $K$ itself. The special case when $K=\FF$ yields the finiteness of the residue field of $\FpXCompl$ mentioned in~\S\ref{subsec:eq_char}.
\subsection{Mixed characteristic}\label{subsec:mixed_char}
The main formalization challenge we face when formalizing the definition of mixed characteristic local fields is analogous to the issue discussed in Remark~\ref{rmk:laurent_vs_adic}. The basic API for the $p$-adic numbers $\QQ_p$ was already available in \mathlib, but it predated the formalization of adic valuations. Since we want to take advantage of this more general theory, our approach is to define a new type $\padicCompl$\href{https://github.com/mariainesdff/local_fields_journal/blob/0b408ff3af36e18f991f9d4cb87be3603cfc3fc3/src/padic_compare.lean#L95}{\extlink}
\begin{lstlisting}
def Q_p : Type* := 
adic_completion $\QQ$ (p_height_one_ideal p)
\end{lstlisting}
and to prove that it is isomorphic, as a valued field, to the field $\QQ_p$. This isomorphism is established in the definition \code{padic_equiv}\href{https://github.com/mariainesdff/local_fields_journal/blob/0b408ff3af36e18f991f9d4cb87be3603cfc3fc3/src/padic_compare.lean#L274}{\extlink} and its construction follows the main strategy explained in~\S\ref{subsubsection:isom_laurent}. Namely, we provide two abstract completions \code{padic_pkg} and \code{padic_pkg'} of $\QQ$ and we upgrade the equivalence as uniform spaces to an isomorphism of valued fields. As a consequence, the unit ball $\padicComplInt$, called\href{https://github.com/mariainesdff/local_fields_journal/blob/0b408ff3af36e18f991f9d4cb87be3603cfc3fc3/src/padic_compare.lean#L332}{\extlink} \code{(Z_p p)} in our code,
is proved to be isomorphic to the $p$-adic integers $\ZZ_p$ in the declaration\href{https://github.com/mariainesdff/local_fields_journal/blob/0b408ff3af36e18f991f9d4cb87be3603cfc3fc3/src/padic_compare.lean#L573}{\extlink}
\begin{lstlisting}[caption={The isomorphism between $\padicComplInt$ and $\ZZ_p$.}, label={code:iso_Zp}]
def padic_int_ring_equiv : (Z_p p) $\simeq $+* $\ZZ$_[p]
\end{lstlisting}

We are now ready to give the following definition:
\begin{definition}\label{def:mixed_char_local_field}
A \emph{mixed characteristic local field} is a finite dimensional field extension of the field $\padicCompl$, for some $p$.
\end{definition}
This is implemented as\href{https://github.com/mariainesdff/local_fields_journal/blob/0b408ff3af36e18f991f9d4cb87be3603cfc3fc3/src/mixed_characteristic/basic.lean#L41}{\extlink}
\begin{lstlisting}
class mixed_char_local_field (p : $\NN$)
  [nat.prime p] (K : Type*) [field K]
  extends algebra (Q_p p) K :=
[to_finite_dimensional : finite_dimensional
  (Q_p p) K]    
\end{lstlisting}
In particular, $\padicCompl$ is a mixed characteristic local field\href{https://github.com/mariainesdff/local_fields_journal/blob/0b408ff3af36e18f991f9d4cb87be3603cfc3fc3/src/mixed_characteristic/basic.lean#L128}{\extlink}.

We formalize the proof that any mixed characteristic local field $K$ is a local field as in Definition~\ref{def:local_field}. The proof is analogous to the one in the equal characteristic case: the only  difference is that every mixed characteristic local field $K$ is automatically separable over $\padicCompl$, since this holds for every algebraic extension of a field of characteristic $0$. Hence, we do not need to assume separability. Finally, the ring isomorphism in the \LClistingname~\ref{code:iso_Zp} implies that the residue field of $\padicComplInt$ is isomorphic to $\FF$, since this is the case for $\ZZ_p$. As above, this ensures that every mixed characteristic local field is indeed a local field according to Definition~\ref{def:local_field}:\href{https://github.com/mariainesdff/local_fields_journal/blob/0b408ff3af36e18f991f9d4cb87be3603cfc3fc3/src/local_field.lean#L74}{\extlink}
\begin{lstlisting}
def mixed_char_local_field.local_field : local_field K := 
{ complete := mixed_char_local_field.complete_space p K,
  is_discrete := v.valuation.is_discrete p K,
  finite_residue_field := ...,
  ..(mixed_char_local_field.valued p K) }
\end{lstlisting} 

The extension of the $p$-adic valuation to $K$\href{https://github.com/mariainesdff/local_fields_journal/blob/0b408ff3af36e18f991f9d4cb87be3603cfc3fc3/src/mixed_characteristic/valuation.lean#L67}{\extlink} is the unique nontrivial discrete\href{https://github.com/mariainesdff/local_fields_journal/blob/0b408ff3af36e18f991f9d4cb87be3603cfc3fc3/src/mixed_characteristic/valuation.lean#L73}{\extlink} valuation on $K$, and the field $K$ is complete with respect to the induced topology\href{https://github.com/mariainesdff/local_fields_journal/blob/0b408ff3af36e18f991f9d4cb87be3603cfc3fc3/src/mixed_characteristic/valuation.lean#L70}{\extlink}.

The \emph{ring of integers} $\rint{K}$ of a mixed characteristic local field $K$ is the integral closure of $\ball{(\padicCompl)}$ in $K$\href{https://github.com/mariainesdff/local_fields_journal/blob/0b408ff3af36e18f991f9d4cb87be3603cfc3fc3/src/mixed_characteristic/basic.lean#L62}{\extlink}. 
\begin{lstlisting}
def ring_of_integers :=
integral_closure (Z_p p) K
\end{lstlisting}

As in the equal characteristic case, we have that $\rint{K} = \ball{K}$, so in particular $\rint{K}$ is a discrete valuation ring\href{https://github.com/mariainesdff/local_fields_journal/blob/0b408ff3af36e18f991f9d4cb87be3603cfc3fc3/src/mixed_characteristic/valuation.lean#L77}{\extlink}.

\section{Discussion}\label{sec:conclusion}

\subsection{Future work}\label{subsec:future_work}

Our next project is to prove that every local field is either a mixed characteristic local field or an equal characteristic local field, and that completions of global fields at finite places are local fields. We will then relate unramified extensions of local fields with extensions of their residue fields, showing that they are all (pro-)cyclic. This paper describes one of the first steps in a larger scale project aiming at formalizing local class field theory. We plan to formulate it in cohomological terms, relying on the recent work~\cite{Liv23} by Livingston.

\subsection{Related works}\label{subsec:related_work}
We formalize our work on top of the \lean library \mathlib and the \lean project \cite{deF23} formalizing extensions of norms. The basic theory of DVR's was available in \mathlib at the start of our project. The library includes a formalization of the additive valuation on a DVR but we discuss in~\S\ref{subsec:implementation} our choice of working with $\wzmZZ$-valued valuations instead.

As far as other systems are concerned, the first formalization of the $p$-adic numbers appeared in the Coq UniMath library in~\cite{padicsCoq}. Beyond the choice of the theorem prover, the two major differences with our work come from different axiomatization settings: the authors of~\cite{padicsCoq} assume the univalence axiom and work in a constructive setting, replacing the notion of being non-zero with a property of being ``apart from zero''. As a consequence, their construction of the $p$-adic integers $\ZZ_p$ follows a completely different path and involves a ``carrying'' operator on the ring $\powerseries[\ZZ]$ and an equivalence relation based on univalent ``paths'' rather than equality. The field $\QQ_p$ is then defined as the Heyting field of fractions of $\ZZ_p$, a construction mimicking the usual field of fractions of an integral domain but in the setting of apartness domains. No treatment of the $p$-adic valuation, of the $p$-adic metric and more generally of topological properties is presented \ibid. From the algebraic point of view, neither the properties of being a DVR, a Dedekind domain, or a local ring, nor algebraic extensions of $\QQ_p$, are approached.

The $p$-adic numbers were later formalized in Isabelle/HOL, whose main library also contains a formalization of formal Puiseux series. The formalization of $\ZZ_p$ in~\cite{padicIntIsa} follows the classical path and is quite complete: both the definition as inverse limit and as completion of $\ZZ$ with respect to the $p$-adic valuation are formalized, together with a proof of their equivalence and of Hensel's lemma. The basic results of the topological properties of $\ZZ_p$ can also be found \ibid, and the paper~\cite{padicIsa} contains deeper results: the field $\QQ_p$ is defined as the fraction field of the ring $\ZZ_p$ formalized in the previous work, and it is endowed with both a valuation and a norm extending the previous ones on $\ZZ_p$. Again, the main topological properties both of $\QQ_p$ and of $\QQ_p^n$ are studied. Nevertheless, neither the structure of $\ZZ_p$ as a DVR, or as a Dedekind domain, are addressed; also, there is no treatment of finite extensions of $\QQ_p$ or of localization results.

The work~\cite{PuiseuxSeriesIsa} presents the formalization of the ring $\puiseuxseries$ of Puiseux series (over any commutative ring $R$), that are a generalization of Laurent series. Both $\powerseries[R]$ and $\laurentseries[R]$ are defined and studied \ibid, and are endowed with a valuation and a norm, whose basic properties are formalized and extended to
\[
\puiseuxseries=\bigcup_{d\geq 1}R(\!(X^{1/d})\!).
\]
The main focus of the paper is the formalization of the Newton--Puiseux' theorem stating that whenever $C$ is an algebraically closed field of characteristic $0$, the same holds for $\puiseuxseries[C]$; as a consequence, few algebraic properties both of $\laurentseries[\FF]$ and of $\puiseuxseries[\FF]$ are formalized. In particular, no formalization of the DVR structure of $\powerseries[\FF]$ or of field extensions of $\laurentseries[\FF]$ can be found \ibid.

\subsection{Remarks about the Implementation}\label{subsec:implementation}

\paragraph{Additive and multiplicative valuations}
The \mathlib library prioritizes multiplicative valuations over additive ones, providing a much wider API for the first ones. For example, the general theory of valuation rings\href{https://leanprover-community.github.io/mathlib_docs/ring_theory/valuation/valuation_ring.html#valuation_ring.valuation}{\extlink} is framed in terms of multiplicative valuations. Moreover, while \mathlib does not yet include the definition of discrete valuation, it does provide specific examples, such as adic valuations on Dedekind domains, which take values in $\wzmZZ$.

On the other hand, \mathlib only provides the formalization of the \emph{additive} valuation \code{add_val}\href{https://leanprover-community.github.io/mathlib_docs/ring_theory/discrete_valuation_ring/basic.html#discrete_valuation_ring.add_val}{\extlink} on a discrete valuation ring, taking values in \code{part_enat}, a decidable version of the type \code{enat}=\,$\NN\mcode{\cup \,\{\infty\}}$. 
For consistency with the rest of the library, we propose to replace \code{add_val} with our implementation of the multiplicative valuation.

\paragraph{Uniformizers}
To indicate that a term \code{(}$\pi$\code{ : R)} in a ring $R$ is a uniformizer for a valuation $v_R$, we provide a predicate \code{is_uniformizer}\href{https://github.com/mariainesdff/local_fields_journal/blob/0b408ff3af36e18f991f9d4cb87be3603cfc3fc3/src/discrete_valuation_ring/basic.lean#L137}{\extlink}:
\begin{lstlisting}
def is_uniformizer ($\pi$ : R) : Prop := 
vR $\pi$ = (of_add (- 1 : $\ZZ$) : $\wzmZZ$)
\end{lstlisting}

We also provide a bundled version of this definition, called \code{uniformizer}\href{https://github.com/mariainesdff/local_fields_journal/blob/0b408ff3af36e18f991f9d4cb87be3603cfc3fc3/src/discrete_valuation_ring/basic.lean#L148}{\extlink}. That is, we provide a structure \code{uniformizer}\relax whose terms consists of two fields: an element $\pi$ of the ring~$R$, together with the proof that $\pi$ is a uniformizer for a given valuation. Since any uniformizer is necessarily a member of the unit ball, we decided to make the field $\pi$ in this definition a term of type \code{vR.integer} (as opposed to type \code{R}).
\begin{lstlisting}
structure uniformizer :=
($\pi$ : vR.integer)
(valuation_eq_neg_one : is_uniformizer vR $\pi$)
\end{lstlisting}

The proposition \code{is_uniformizer} is useful when proving lemmas that apply to any uniformizer element, such as the lemma stating that every uniformizer is nonzero\href{https://github.com/mariainesdff/local_fields_journal/blob/0b408ff3af36e18f991f9d4cb87be3603cfc3fc3/src/discrete_valuation_ring/basic.lean#L178}{\extlink}:
\begin{lstlisting}
lemma uniformizer_ne_zero {$\pi$ : R}
  (h$\pi$ : is_uniformizer vR $\pi$) : $\pi \ne 0$ := 
\end{lstlisting}
By contrast, the bundled definition is more convenient to prove results that involve fixing a uniformizer. For example, we use it when proving that every nonzero $r : \ball{K}$ can be factored as the product of a unit by a power of a fixed uniformizer\href{https://github.com/mariainesdff/local_fields_journal/blob/0b408ff3af36e18f991f9d4cb87be3603cfc3fc3/src/discrete_valuation_ring/basic.lean#L259}{\extlink}:
\begin{lstlisting}
lemma pow_uniformizer {r : $K_0$} (hr : r $\ne$ 0)
  ($\pi$ : uniformizer v) :
  $\exists$ n : $\mathbb{N}$, $\exists$ u : $K_0^\times$, r = $\pi$.1^n * u :=
\end{lstlisting}

For an in-depth discussion of bundled versus unbundled representations in Lean and in Coq, we refer the reader to \cite{Baa22} or \cite{typeclassesCoq}.

\paragraph{Extensions and \code{valued} instances}
In \S\ref{subsec:extensions} we construct, for each field $K$ complete with respect to a discrete valuation~$v$ and every finite extension $L/K$, a unique valuation $v_L$ on $L$, recorded as \code{(v_L : valuation L}$\;\wzmZZ)$. However, at this level of generality we do not put a global \code{valued} instance on $L$: doing so would create infinitely many ``diamonds'', by which we mean the existence of two different procedures to endow an object with a certain mathematical structure (represented by two terms of the given structure that are propositionally, but not definitionally, equal). This leads to problems when exploiting the full force of the typeclass inference mechanism: indeed, suppose that in a given context $\Gamma$, a term $t:T$ is well-typed assuming the existence of a term \code{s:S} for a certain structure \code{S}. Given two terms \code{s}$_1$,\code{s}$_2$\code{:S} that are propositionally equal, the corresponding terms \code{t}$_1$,\code{t}$_2$\code{:T} are different, yet carry the same mathematical information. Now, if the type-class inference mechanism can produce both terms \code{s}$_1$ and \code{s}$_2$ --- say, inside a proof --- the user would face the problem of having two different terms that look indistinguishable, which clearly leads to unexpected problems. In our setting, for every field $K$ that is complete with respect to a discrete valuation $v$, the original valuation is propositionally but not definitionally equal to its trivial extension $v_K$, resulting in two terms of the structure of a \code{valued} field on $K$: to avoid any diamond, we need to exclude the second from the inference search, and therefore we need to avoid declaring a structure of \code{valued} field on the trivial extension $K$ of $K$: this forces us to avoid declaring such a structure on \emph{every} extension $L/K$. Nevertheless, we provide a lemma \code{trivial_extension_eq_valuation}\href{https://github.com/mariainesdff/local_fields_journal/blob/0b408ff3af36e18f991f9d4cb87be3603cfc3fc3/src/discrete_valuation_ring/trivial_extension.lean#L63}{\extlink} proving the equality $v=v_K$.

We make an exception to this rule when implementing mixed and equal characteristic local fields, since we do declare \code{valued} instances for them. This creates a mild diamond for $\padicCompl$ and $\FpXCompl$ but this does not cause trouble in our formalization, thanks to the comparison lemma \code{trivial_extension_eq_valuation}\href{https://github.com/mariainesdff/local_fields_journal/blob/0b408ff3af36e18f991f9d4cb87be3603cfc3fc3/src/discrete_valuation_ring/trivial_extension.lean#L63}{\extlink} mentioned above. Indeed, in all cases where a the inference system would create problems, the lemma allows us to disambiguate the situation: for an example, the reader can inspect the proof of \code{FpX_int_completion.equiv_valuation_subring}\href{https://github.com/mariainesdff/local_fields_journal/blob/0b408ff3af36e18f991f9d4cb87be3603cfc3fc3/src/eq_characteristic/valuation.lean#L102}{\extlink}, proving that the unit ball $\FpXComplInt$ is isomorphic to the subring of $\FpXCompl$ whose elements have $(X)$-adic valuation less than or equal to (\code{1:} $\wzmZZ$).

\paragraph{Fraction fields.}
If $K$ is a discretely valued field, then $\ball{K}$ is a DVR, so we put a \code{[valued (fraction_ring K}$_0$) $\wzmZZ$\code{]} instance on $\Frac(\ball{K})$\href{https://github.com/mariainesdff/local_fields_journal/blob/0b408ff3af36e18f991f9d4cb87be3603cfc3fc3/src/discrete_valuation_ring/basic.lean#L458}{\extlink}. Note that, while $K$ and $\Frac(\ball{K})$ are isomorphic, they are represented by different types in \mathlib, so the above \code{valued} instance has a different type from the original \code{[valued K} $\wzmZZ$\code{]} instance and no diamond occurs. However, it would if we had instead decided to put this \code{valued} instance on any field $L$ satisfying the condition \code{[is_fraction_ring K}$_0\; $\code{L]}, since this holds for $K$.

\paragraph{The \code{normed_field} and the \code{valued} instances}
Recall from~\S\ref{subsec:extensions} that to each discrete valuation $v$ on a field $K$ we can associate a nonarchimedean multiplicative norm $\lvert\cdot\rvert_K$.

When there exists a preferred discrete valuation $v$ on $K$, we often register it as a \code{valued} instance:
\begin{lstlisting}
    {K : Type*} [field K] [hv : valued K $\wzmZZ$] 
\end{lstlisting}
In this situation, the norm $\lvert\cdot\rvert_K$ associated to $v$ would be the preferred nontrivial norm on $K$ , up to rescaling, and we would like to record a corresponding \code{normed_field} instance on $K$:
\begin{lstlisting}
    {K : Type*} [normed_field K $\wzmZZ$] 
\end{lstlisting}
Note that the datum of a field is embedded into the definition of the \code{normed_field} class, while the class \code{valued} takes the field structure as an argument. It follows that declaring a \code{normed_field} instance on every field that carries a \code{valued} one leads to a loop in the typeclass inference system\footnote{\lean[4] can detect these kinds of simple loops, so this should not be an issue once our project has been ported to \lean[4].}. We discuss below a concrete example to illustrate this problem, for which we studied
a trace of the instance search using the option \code{trace.class_instances} to track the typeclass inference process.

Suppose that we had declared \code{discretely_normed_field} as an instance instead of as a definition, so that every discretely valued field
\begin{lstlisting}
(K : Type*) [field K] [hv : valued K $\wzmZZ$]
  [is_discrete hv.v]
\end{lstlisting}
would automatically inherit a \code{normed_field} instance.
Under this assumption, the typeclass inference system would get into an infinite loop when searching for an instance of
\code{valuation_ring hv.v.valuation_subring.to_subring,}\relax 
because one of the first results that Lean tries to apply is \code{valuation_ring.of_field}, that automatically infers a valuation ring structure on every field. In turn, this reduces the problem at hand to finding an instance of \code{field hv.v.valuation_subring.to_subring.}
This leads to a dead end (since the valuation subring is not a field); however, while trying to find this instance, the inference system finds \code{normed_field.to_field}\relax, and hence it starts to search for an instance of
\code{normed_field hv.v.valuation_subring.to_subring}, for which it will try to apply the definition \code{discretely_normed_field}, that as input requires a \code{field} instance on the valuation subring which recreates the problem discussed above, resulting in the system getting stuck in an infinite loop.

However, observe that no problem arises when putting both a \code{valued} and a \code{normed_field} instance on, for instance, $\QQ_p$, because the typeclass inference system is able to find these instances only on $\QQ_p$.

\paragraph{Laurent series}
As for general DVR's, an additive valuation was already available in \mathlib for power series, with the name \code{hahn_series.add_val}\href{https://leanprover-community.github.io/mathlib_docs/ring_theory/hahn_series.html#hahn_series.add_val}{\extlink}. It is \code{part_enat}-valued and it is defined as the greatest $n$ such that $X^n$ divides the power series. Although some basic API was available, the same reasons that led us to systematically work with multiplicative valuations rather than additive ones pushed us to rely on the general theory of $\wzmZZ$-valued adic valuations rather than with this \emph{ad hoc} definition.

The main isomorphism \code{laurent_series_ring_equiv} exhibited in the \LClistingname~\ref{code:isom_ls} is defined as \emph{the inverse} of an isomorphism\href{https://github.com/mariainesdff/local_fields_journal/blob/0b408ff3af36e18f991f9d4cb87be3603cfc3fc3/src/laurent_series_equiv_adic_completion.lean#L828}{\extlink}
\begin{lstlisting}
ratfunc_adic_compl_ring_equiv : $ \FpXCompl[K]\simeq ^{ +*}\laurentseries[K]$.
\end{lstlisting}
The reason is that to prove additivity and multiplicativity of the above map it suffices to observe that it coincides with the extension $\compl{\mcode{coe}}$ (in the sense of~\eqref{eq:univ_compl}) of the coercion
\begin{lstlisting}
            coe : $K(X)\rightarrow^{+*} \laurentseries[K]$.
\end{lstlisting}
The formalism of uniform completions suffices to establish that $\compl{\mcode{coe}}$ is a ring homomorphism simply because $\mcode{coe}$ is, whereas there exists no explicit ring homomorphism $\varphi$ such that $\compl{\varphi}$ \code{= laurent_series_ring_equiv}.
 
\paragraph{The field of $p$-adic numbers as adic completion}
Although the formalization of the isomorphism $\padicCompl\cong\QQ_p$ follows in many respects the one for Laurent series, one notable difference is that the type $\QQ$, unlike $K(X)$, already bore an instance of \code{metric_space}, induced from the euclidean distance, and this induced a uniform structure on $\QQ$. In order to define the term $\padicCompl$ we needed to access the API concerning completions of adic spaces, and to do so a \code{valued} instance needed to be defined on $\QQ$. The corresponding uniform space structure would conflict with the euclidean one, and therefore we needed to locally disable the \code{metric_space} instance on $\QQ$ already just to \emph{define} the type $\padicCompl$.

\begin{acks}
We would like to thank Riccardo Brasca, Mario Carneiro, Antoine Chambert-Loir and Rob Lewis for interesting discussions about possible formalizations of the $p$-adic numbers, Heather Macbeth for helpful conversations about the formalization of Laurent series, Yaël Dillies for providing access to some results about well-founded relations and Cyril Cohen for answering some questions about $p$-adic numbers in Coq. We also thank the \mathlib community as a whole for their support during the completion of this project.

The first author acknowledges support from the Grant CA1/RSUE/2021-00623 funded by the Spanish Ministry of Universities, the Recovery, Transformation and Resilience Plan, and Universidad Autónoma de Madrid. She is also grateful for the one-month invitation from Université Jean Monnet Saint-Étienne, as part of the
program "Appel à Projets Recherche UJM 2023".

This work was supported by the LABEX MILYON (ANR-10-LABX-0070) of Université de Lyon, within the program ``France 2030'' (ANR-11-IDEX-0007) operated by the French National Research Agency (ANR).
\end{acks}
\balance

\bibliographystyle{ACM-Reference-Format}
\bibliography{CPP_biblio}


\end{document}

%% file: CPP_preamble.tex
\usepackage[T1]{fontenc}
\definecolor{Burgundy}{RGB}{128,0,32}
\usepackage{enumitem}
\usepackage{amsfonts, amsthm, amsmath}
\usepackage{soul}
\usepackage[english]{babel}
\usepackage{mathrsfs}
\usepackage{fontawesome} 
\usepackage{xpunctuate}
\usepackage{hyperref}
\usepackage{microtype}
 \usepackage{balance}
\newcommand{\extlink}{~\ensuremath{{}^\text{\faExternalLink}}}
\newcommand*{\ibid}{\emph{ibid}\xperiod}
\usepackage{listings}

\renewcommand{\lstlistingname}{Code excerpt}
\newcommand{\LClistingname}{\MakeLowercase{\lstlistingname}}
\newcommand{\code}[1]{\lstinline{#1}}
\newcommand{\mcode}[1]{\mathtt{\ #1\ }}
\newcommand*{\lean}[1][3]{Lean~{#1}\xspace}
\newcommand*{\mathlib}{\texttt{mathlib}\xspace}
\theoremstyle{plain}
\newtheorem{theorem}{Theorem}[section]

\theoremstyle{definition}
\newtheorem{definition}[theorem]{Definition}

\theoremstyle{remark}
\newtheorem{remark}[theorem]{Remark}
\newlist{listResults}{enumerate}{1}
\setlist[listResults]{label={\arabic*}.}
\newlist{sublistResults}{enumerate}{1}
\setlist[sublistResults]{label={\arabic{listResultsi}}-\,\alph*\textup{)}}
\newlist{listConditions}{enumerate}{1}
\setlist[listConditions]{label=\textup{({\Alph*})}}
\newlist{listDef}{enumerate}{1}
\setlist[listDef]{label={\roman*}\textup{)}}
\newlist{listHyp}{enumerate}{1}
\newcommand*{\NN}{\mathbb{N}}
\newcommand*{\ZZ}{\mathbb{Z}}
\newcommand*{\wzmZZ}{\ZZ_{\mathtt{m}0}} 
\newcommand*{\RR}{\mathbb{R}}
\newcommand*{\QQ}{\mathbb{Q}}
\newcommand*{\CC}{\mathbb{C}}
\newcommand*{\FF}[1][p]{\mathbb{F}_{#1}}
\newcommand*{\m}{\mathfrak{m}}
\newcommand*{\av}[1][\empty]{a_{#1}}
\DeclareMathOperator{\Frac}{Frac} 
\newcommand*{\laurentseries}[1][{\FF}]{\ensuremath{#1(\!(X)\!)}} 
\newcommand*{\puiseuxseries}[1][{R}]{\ensuremath{#1\{\!\{X\}\!\}}} 
\newcommand*{\compl}[2][\empty]{\ensuremath{\widehat{#2}_{#1}}} 
\newcommand*{\FpXCompl}[1][{\FF}]{\compl{#1(X)}}  
\newcommand*{\FpXComplInt}[1][\FF]{\ball{\bigl(\FpXCompl[#1]\bigr)}}  
\newcommand*{\padicCompl}{\compl[(p)]{\QQ}}  
\newcommand*{\padicComplInt}{\ball{\bigl(\padicCompl\bigr)}}  
\newcommand*{\powerseries}[1][{\FF}]{#1{[\![X]\!]}} 
\newcommand*{\maxid}[1][R]{\mathfrak{m}_{#1}} 
\newcommand*{\rint}[1]{\mathcal{O}_{#1}} 
\newcommand*{\normalised}[1][v]{{#1}_{0}} 
\newcommand*{\ball}[1]{{#1}^\circ} 
\newcommand*{\filter}{\mathscr{F}}
\newcommand*{\princfilter}{\boldsymbol{\mathcal{P}}}

\DeclareMathOperator{\image}{Im} 